\documentclass[11pt]{article}

\usepackage{geometry,fullpage}
\usepackage{titling}
\usepackage{bm}
\usepackage{amsthm,amssymb,amsmath,amsfonts,rotate}
\usepackage{graphicx,algorithmic,algorithm}
\usepackage[]{datetime}
\usepackage{aliascnt}

\usepackage[pdftex, plainpages = false, pdfpagelabels,
bookmarks=false,
bookmarksopen = true,
bookmarksnumbered = true,
breaklinks = true,
pagebackref,
colorlinks = true,
linkcolor = blue,
urlcolor = magenta,
citecolor = red,
anchorcolor = green,
hyperindex = true,
hyperfigures
]{hyperref}

\usepackage{todonotes}
\usepackage{lmodern}
\usepackage{tikz}
\usetikzlibrary{calc,3d}
\usetikzlibrary{intersections,decorations.pathmorphing,shapes,decorations.pathreplacing,fit}
\usepackage{amsmath,ifthen,eurosym}
\usepackage{latexsym,amssymb,url}
\usepackage{subcaption}
\usepackage{cleveref}

\usepackage{amsmath,latexsym,amsthm,amsfonts,amssymb,%
	 ifthen,color,wrapfig,hyperref,rotate,lmodern,aliascnt,%
	 datetime,graphicx,algorithmic,algorithm,enumerate,enumitem,%
	todonotes,cleveref,bm,subcaption,tabularx,xcolor,colortbl, xspace%
}

\newcommand{\bd}{{\sf bd}}

\newcommand{\inter}{{\sf int}}
\newcommand{\remove}[1]{}

\newcommand{\bigmid}{\;\big|\;}
\newcommand{\w}{\operatorname{\textsf{w}}}
\newcommand{\cupall}{\pmb{\pmb{\bigcup}}}

\newcommand{\sshow}[2]{\ifthenelse{\equal{#1}{0}}{#2}{}}

\newcounter{func}
\newcommand{\newfun}[1]{f_{\refstepcounter{func}\label{#1}\thefunc}}
\newcommand{\funref}[1]{\hyperref[#1]{f_{\ref*{#1}}}} % print a
\newcounter{con}

\newcommand{\conref}[1]{\hyperref[#1]{c_{\ref*{#1}}}} % print a
\usepackage[T1]{fontenc}
\usepackage[utf8]{inputenc}
\usepackage{alphabeta}

\newcommand{\tw}{{\sf tw}}

\usepackage{dsfont}

\usepackage{float}
\usepackage{tikz,pgf}
\usetikzlibrary{backgrounds}
\usetikzlibrary{calc,3d}
\usepackage{color,cite}

\definecolor{MidnightBlue}{rgb}{0.1,0.1,0.44}
\definecolor{Black}{rgb}{0,0, 0}
\definecolor{Blue}{rgb}{0, 0 ,1}
\definecolor{Red}{rgb}{1, 0 ,0}
\definecolor{White}{rgb}{1, 1, 1}
\definecolor{Grey}{rgb}{.6, .6, .6}
\definecolor{Mygreen}{rgb}{.0, .7, .0}
\definecolor{Yellow}{rgb}{.55,.55,0}
\definecolor{Mustard}{rgb}{1.0, 0.86, 0.35}
\definecolor{applegreen}{rgb}{0.55, 0.71, 0.0}
\definecolor{darkturquoise}{rgb}{0.0, 0.81, 0.82}
\definecolor{celestialblue}{rgb}{0.29, 0.59, 0.82}
\definecolor{green-yellow}{rgb}{0.68, 1.0, 0.18}
\definecolor{crimsonglory}{rgb}{0.75, 0.0, 0.2}
\definecolor{darkmagenta}{rgb}{0.30, 0.0, 0.30}

\newcommand{\change}[1]{{#1}}

\tikzset{red node/.style={draw=red, circle, fill = red, minimum size = 4pt, inner sep = 0pt}}
\tikzset{yellow node/.style={draw=yellow, circle, fill = yellow, minimum size = 4pt, inner sep = 0pt}}
\tikzset{blue node/.style={draw=celestialblue, circle, fill =celestialblue, minimum size = 4pt, inner sep = 0pt}}
\tikzset{triangle/.style = { regular polygon, regular polygon sides=3, rotate=180}}
\tikzset{small red/.style={draw=red, triangle, fill = red, minimum size = 2pt, inner sep = 0pt}}

\tikzset{black node/.style={draw, circle, fill = black, minimum size = 3pt, inner sep = 0pt}}
\tikzset{small black node/.style={draw, circle, fill = black, minimum size = 3pt, inner sep = 0pt}}
\tikzset{model node/.style={draw=celestialblue, circle, fill = celestialblue, minimum size = 5pt, inner sep = 0pt}}
\tikzset{model node small/.style={draw=celestialblue, circle, fill = celestialblue, minimum size = 3pt, inner sep = 0pt}}
\tikzset{rep node/.style={draw=red, circle, fill = red, minimum size = 3pt, inner sep = 0pt}}
\tikzset{track node 1/.style={draw, circle, fill = black, minimum size = 2pt, inner sep = 0pt}}
\tikzset{track node 2/.style={draw=black!30!white, circle, fill = black!30!white, minimum size = 2pt, inner sep = 0pt}}
\tikzset{track node 3/.style={draw=black!10!white, circle, fill = black!10!white, minimum size = 2pt, inner sep = 0pt}}

\newcommand{\mynewtheorem}[2]{
	\newaliascnt{#1}{dummy}
	\newtheorem{#1}[#1]{#2}
	\aliascntresetthe{#1}
}

\theoremstyle{plain}
\mynewtheorem{theorem}{Theorem}
\mynewtheorem{proposition}{Proposition}
\mynewtheorem{corollary}{Corollary}
\mynewtheorem{lemma}{Lemma}

\theoremstyle{definition}
\mynewtheorem{definition}{Definition}
\theoremstyle{remark}
\mynewtheorem{sublemma}{Sublemma}
\mynewtheorem{claim}{Claim}
\mynewtheorem{remark}{Remark}
\mynewtheorem{fact}{Fact}
\mynewtheorem{conjecture}{Conjecture}
\mynewtheorem{question}{Question}
\mynewtheorem{observation}{Observation}
\mynewtheorem{problem}{Problem}
\mynewtheorem{note}{Note}

\newcommand{\frR}{{\frak{R}}}
\newcommand{\Ocal}{{\mathcal O}\xspace}

\hypersetup{
colorlinks = true,
linkcolor = black!30!blue,
citecolor = black!30!green
}

\usepackage{titling}
\thanksmarkseries{arabic}
\newcommand*\samethanks[1][\value{footnote}]{\footnotemark[#1]}

\newtheorem{innercustomgeneric}{\customgenericname}
\providecommand{\customgenericname}{}

\begin{document}
%\removed{
\title{A more accurate view of the Flat Wall Theorem}

\author{\bigskip Ignasi Sau\thanks{LIRMM, Université de Montpellier, CNRS, Montpellier, France. {Supported}  by  the ANR projects DEMOGRAPH (ANR-16-CE40-0028), ESIGMA (ANR-17-CE23-0010), ELIT (ANR-20-CE48-0008), the French-German Collaboration ANR/DFG Project UTMA (ANR-20-CE92-0027), and the French Ministry of Europe and Foreign Affairs, via the Franco-Norwegian project PHC AURORA.  Emails:  \texttt{ignasi.sau@lirmm.fr}, \texttt{giannos.stamoulis@lirmm.fr}, \texttt{sedthilk@thilikos.info}}\and
Giannos Stamoulis\samethanks[1] \and
Dimitrios  M. Thilikos\samethanks[1]}
\date{}

\maketitle

\begin{abstract}
%\autoref{label_technologically}
	\noindent We introduce a supporting combinatorial framework for the Flat Wall Theorem.
	In particular, we suggest two variants of the theorem and we introduce a new, more versatile,  concept
	of wall homogeneity as well as 	
	the notion of regularity in flat walls. All  proposed concepts and results
	 aim at facilitating  the use of the irrelevant vertex technique in  future algorithmic applications.

	\bigskip

\end{abstract}

\newpage

%	
%	{\scriptsize
%\ig{These are some things to unify:}
%\begin{itemize}
%\item macro for treewidth: I prefer ${\sf tw}.$ {OK}
%\item belong to / belong to.
%\item $G - S$ or $G \setminus S$? {$\to$ $G \setminus S$}
%\item Use $\bigO{...}$ instead of $O(...)$ {OK}
%\item In expressions such as $2^{{\sf poly}(k)}n^3,$ I prefer $2^{{\sf poly}(k)} \cdot n^3$ (with dot). {OK}
%\end{itemize}}
%

\newpage
\tableofcontents
\newpage

%\setcounter{page}{1}
%
%\newpage
%
\section{Introduction}
\label{label_technologically}
One of the cornerstone achievements of the Graph Minors series by Robertson and Seymour
was the celebrated  {\em Flat Wall Theorem}, proved in the 13th paper of the series~\cite{RobertsonS95b}. It is a powerful graph structural result, revealing the local structure of $H$-minor-free graphs.
The Flat Wall Theorem has important consequences and  applications in structural graph theory and in graph algorithm design. It served as the combinatorial base for the  design of an algorithm for the following two problems:
\begin{itemize}
\item  {\sc Minor Testing}: Given a graph $G$ and a $k$-vertex graph $H,$ decide  whether $G$ contains $H$ as a minor.
\item {\sc Disjoint Paths}: Given a graph $G$ with $k$ pairs of terminals $(s_{i},t_{i}),\ldots,(s_{k},t_{k}),$ decide  whether $G$ contains $k$ vertex-disjoint paths joining $s_{i}$ and $t_{i}$
for every $i\in \{1,\ldots,k\}.$
\end{itemize}
These algorithms run in time $f(k)\cdot n^3$ on $n$-vertex graphs, for some function $f:\Bbb{N}\to\Bbb{N}$ (see~\cite{KawarabayashiKR11thed} for quadratic-time improvements).
This, using the terminology of parameterized complexity, implies that both above problems,
when parameterized by $k,$ belong to the parameterized class ${\sf FPT}$ or, alternatively, {\em admit ${\sf FPT}$-algorithms}.
In order to obtain these algorithms,  Robertson and Seymour introduced a powerful technique, called the {\em  irrelevant vertex technique}, which has now become a standard technique
in the design of parameterized algorithms (see e.g., Section 7 of the textbook~\cite{CyganFKLMPPS15para}).
Further algorithmic applications combining  the Flat Wall Theorem and the irrelevant vertex technique appeared later  in \cite{AdlerGK08comp,DawarGK07loc,GroheKMW11find,KawarabayashiK10thee,FominLPSZ20hitt}, while generalizations to directed graphs have recently appeared in \cite{HuynhJW19anun,GiannopoulouKKK20thed}.%\medskip

%\ig{please use the unified format for the refs (full names of authors and journals, DOIs, etc...}

\subsection{The Flat Wall Theorem and its variants}

The original statement of the  Flat Wall Theorem, as appeared in~\cite{RobertsonS95b}, is the following.

\begin{proposition}
\label{label_tautologisch}
There exist functions $f:\Bbb{N}^2\to\Bbb{N}$ and $f':\Bbb{N}^2\to\Bbb{N}$ such that
if $G$ is a graph and $h$ and $k$ are integers, then one of the following holds:
\begin{enumerate}
\item $G$ contains $K_{h}$ as a minor\footnote{I.e., some subgraph of $G$ can be contracted to a complete graph on $h$ vertices.}.
\item $G$ has treewidth at most $f(k,h).$
\item $G$ has a vertex set $A$ with $|A|\leq f'(h),$ such that
$G\setminus A$ contains a  flat wall $W$ of height $k.$
\end{enumerate}
\end{proposition}
\noindent We postpone the formal definitions of   ``treewidth'', the related concept of ``tree decomposition'',
%\ig{``tree decomposition'' has not been mentioned so far},
and  ``flat wall'' to \autoref{label_schlussgesetze}.   One can get a quick idea of a  wall by   looking at \autoref{label_thoroughfare} and of flat wall by looking at \autoref{label_consumadamente} and \autoref{label_exhalaciones}.
Intuitively, a flat wall $W$ is contained in a larger graph, its {\sl compass},
that  is separated from the rest of the graph via a separator $S$ that is a ``suitably chosen''
part of  the
perimeter of $W.$ This compass is ``flat'' in the sense that
it does not contain two disjoint paths whose endpoints are in $S$ and are ``crossing''
with respect to the cyclic ordering induced in $S$  by the perimeter of $W.$
As proved by Kawarabayashi, Thomas, and Wollan~\cite{KawarabayashiTW18anew}, this flatness property
can be certified by a concept called  {\sl rendition} (corresponding to the concept of
{\sl rural division}  in \cite{RobertsonS95b}) that can be seen as
a plane embedding inside a disk of a hypergraph with hyperedges of arity at most three (see \autoref{label_operationalization} for a visualization of a rendition).
Then the compass is ``embedded'' inside the rendition so that it can be seen as the union
of graphs called {\sl flaps} bijectively mapped to the hyperedges of the rendition.

In its original version in~\cite{RobertsonS95b}, \autoref{label_tautologisch} was proved
 for $f'(h)=\binom{h}{2}$ with the additional assertion that $f(k,h)$ is  a bound on the treewidth of the ``internal flaps'', i.e., those that are not incident to the perimeter of $W$.
 Later, in   \cite{GiannopoulouT13opti},
 the same result
was proved (without an algorithm)
for $f'(h)=h-5$ and $f(k,h)=\Ocal_{h}(k).$\footnote{The notation `$\Ocal_h(\cdot)$' means that the hidden constants
depend only on $h.$}  The original result of Robertson and Seymour was accompanied with an  ${\cal O}(n\cdot m)$-time algorithm\footnote{In this paper we always denote by $n$ and $m$ the number of vertices and edges, respectively, of the graph under consideration.
} that outputs a certifying structure for each possible outcome. This algorithm was further improved to a linear one  by Kawarabayashi,  Kobayashi, and  Reed in~\cite{KawarabayashiKR11thed}.
%\ig{shouldn't we mention author name's here as well, as in the other articles?}.
\smallskip

A recent wave of improvements of \autoref{label_tautologisch} appeared in the following form~\cite{Chuzhoy15impr,KawarabayashiTW18anew}. %\ig{we should say who proved the following proposition before stating it}

\begin{proposition}
\label{label_particularly}
There exist functions $f:\Bbb{N}^2\to\Bbb{N}$ and $f':\Bbb{N}^2\to\Bbb{N}$ such that
if $G$ is a graph and $h$ and $k$ are integers, and $G$ contains a wall $W$ of height $f(k,h)$
as a subgraph, then one of the following holds:
\begin{enumerate}
\item $G$ contains $K_{h}$ as a minor.
\item $G$ has a vertex set $A$ with $|A|\leq f'(h),$ such that
$G\setminus A$ contains a  flat wall $W'$ of height $k.$
\end{enumerate}
\end{proposition}

Notice that \autoref{label_particularly} can indeed be seen as an extension of \autoref{label_tautologisch} because the exclusion of a wall of height $k$ in a $K_{h}$-minor-free graph implies that its treewidth is bounded by $\Ocal_{h}(k)$ \cite{DemaineH08line,KawarabayashiK20line}.
Moreover, according to \cite[Theorem 1.9]{KawarabayashiTW18anew}, \autoref{label_particularly} holds for $f'(h)=\Ocal(h^{24})$
and $f(k,h) = \Ocal(h^{24}(h^2 + r))$, and it enjoys the following  additional features:
\begin{itemize}
\item[(A)] In the first case, the graph $K_{h}$ is a minor
of $G$ in a way that is ``grasped by the wall $W$'.'\footnote{We avoid here the
formal definition of  ``grasping by a wall'' as we do not make use of it in this paper; see~\cite{KawarabayashiTW18anew} for the details. However, we stress that
it provides additional information that is used in further applications (see e.g., \cite{KawarabayashiTW21quic}).}
\item[(B)]  In the second case, the  flat wall $W'$ is  a {\sl subwall} of $W.$
\item[(C)]  \autoref{label_particularly} comes with an algorithm
that certifies one of the two outcomes in  linear time, in particular, in   $\Ocal(h^{24}\cdot m +n)$ time.
\end{itemize}
%
%According to~\cite[Theorem 1.5]{KawarabayashiTW18anew}, \autoref{label_particularly} with Features (A) and (B) \ig{what about (C)?}
%holds for $f'(h)=\Ocal(h^{24})$
%and $f(k,h) = \Ocal(h^{24}(h^2 + r)).$
Moreover, the same result with Features (A) and (B) is proved in \cite[Theorem 1.7]{KawarabayashiTW18anew}  with  the optimal function $f'(h)=h-5$ at the cost of a worst bound for $f(k,h).$ Also \cite[Theorem 1.8]{KawarabayashiTW18anew}  corresponds to
 \autoref{label_particularly} with the additional feature that the compass of the flat wall $W'$ contains no wall of height $f(k,h)+1,$ again at the cost of a worst bound for $f(k,h).$

Later, Chuzhoy~\cite{Chuzhoy15impr} drastically  improved the bounds of  \autoref{label_particularly} with the extra Features (A) and (B) to  $f'(h)=h-5$
and $f(k,h) = \Ocal(h\cdot (h+k)).$
Moreover,  Chuzhoy gives  a polynomial-time algorithm for her improved variant, however she does not specify whether this algorithm  can run in linear time,
 as the one in~\cite[Theorem 1.9]{KawarabayashiTW18anew}.

\subsection{Our contribution}
In this paper we provide a series of enhanced algorithmic
versions of the Flat Wall Theorem as well as a series of combinatorial tools related to
the applicability of the  irrelevant vertex technique.
In our presentation we
adopt the framework and the
terminology of \cite{KawarabayashiTW18anew}. Our aim is to
 introduce a ``more accurate'' view of the Flat Wall Theorem
that, we hope, will be useful for future algorithmic applications.  Our contribution consists in the following.

\paragraph{($α$) Subwalls of flat walls are not always flat.} Our initial motivation comes from the fact\footnote{This was first spotted in the conference article~\cite{SauST20anfp}.} that the  claimed
Feature (B) in \autoref{label_particularly}, as stated in~\cite{KawarabayashiTW18anew}, needs some slight (but not neglectable) revision.
 This feature
is based on   \cite[Lemma 6.1]{KawarabayashiTW18anew}, asserting that
if $W$ is a flat wall and $W'$ is a subwall of $W$ that is disjoint from the perimeter of $W,$
then $W'$ is {\sl also} a flat wall of $G.$ As we observe in~\autoref{label_prognostication},  there
are some very marginal cases where a subwall of a flat wall is not flat anymore.
This phenomenon is illustrated in the flat wall of \autoref{label_consumadamente} (in~\autoref{label_prognostication}).

\paragraph{($\beta$) A reparation framework.} Fortunately, the issue raised in {\bf ($α$)} is just a minor formal mismatch
that harms neither the spirit of the
proofs of \cite{KawarabayashiTW18anew} nor the ``essential'' correctness of
 subsequent results that are based on  \cite{KawarabayashiTW18anew}.
 %, such as~\cite{KawarabayashiTW21quic}.
The first contribution of  our paper is to propose an extension of the framework of \cite{KawarabayashiTW18anew}
that supports a formally correct statement of  Feature (B) in \autoref{label_particularly}.
What we show (\autoref{label_proporcionada}) is that if a wall $W$ is a flat wall, whose flatness is certified by some rendition $\frak{R},$ and $W'$ is a subwall of $W,$
then there is another, slightly different, subwall $\tilde{W}'$ of $W,$ which we call a {\sl $W'$-tilt}, that is indeed flat\footnote{In fact,  \autoref{label_proporcionada} applies not only to subwalls $W'$ of $W$, but also to every subwall $W'$ of the compass of $W$ that is not ``contained'' in a flap. See \autoref{label_prognostication} for the details.}.
By the term ``slightly different'' we mean that $W'$ and its $W'$-tilt $\tilde{W}'$ may differ
{\sl only} perimetrically. Moreover, the rendition certifying the flatness of $\tilde{W}'$ maintains
all the ``internal'' structure of the rendition $\frak{R},$ relatively to $W'.$
This implies that all the arguments based on \autoref{label_particularly} of \cite{KawarabayashiTW18anew} are essentially  correct, and can become
formally correct under the suggested framework.
In our definitions and proofs
we pay attention to all the necessary formalism so to facilitate dealing with future results that
may use those of \cite{KawarabayashiTW18anew} (or \cite{Chuzhoy15impr}).
We conclude with \autoref{label_inconsiderable} that is a version of \autoref{label_particularly} translated into our framework.
%In fact, this version is stronger

\paragraph{($γ$)
A Flat Wall Theorem with compasses of bounded treewidth.} Our next result, \autoref{label_proletarians} in~\autoref{label_distrettamente}, is an improved version of \autoref{label_tautologisch} with the following additional features:
(1) $f(k,h)=k\cdot 2^{\Ocal(h^2 \log h)}$ and $f'(h)=\Ocal(h^{24}),$
(2) in the third case, the compass of the wall $W$ comes with a tree decomposition of width at most $f(k,h),$ and
(3) the result is accompanied by  a $2^{{\cal O}_{h}(r^2)}\cdot n$ time algorithm.
Notice that a non-algorithmic version of  this  result could be indirectly derived, with worst functions,
combining \cite[Theorem 1.8]{KawarabayashiTW18anew} and the main result of Kawarabayashi and Kobayashi in~\cite{KawarabayashiK20line}. We present this result in this paper for the following reasons:
first because it is new, second because it  is in a form suitable for future applications where it is important that the compass
has bounded treewidth, and third because its proof provides an indicative sample of the potential of
the formalism  of $W'$-tilts that we suggest in {\bf ($β$)}.

\paragraph{($δ$)  An alternative concept of wall homogeneity.}
As mentioned before, the Flat Wall Theorem has been the combinatorial base for the {\sf FPT}-algorithms of \cite{RobertsonS95b}
for {\sc Minor Testing} and {\sc Disjoint Paths}. One of the cornerstone  ideas of \cite{RobertsonS95b}
was to prove that the existence of a ``big enough'' flat wall $W$
in the input graph $G$ implies that a minor-model of $H$ or a collection of $k$ disjoint paths
in $G$ can be safely rerouted so to avoid the central vertices of this wall (see \autoref{label_thoroughfare}
for a visualization of the central vertices of a wall). This permits us to declare parts of the wall ``irrelevant''
and find an equivalent instance of the problem with fewer vertices.
In fact, avoiding the central vertices is not so straightforward when dealing with a flat wall $W.$
This is because the rerouting has to be done inside the compass $K$
of $W$ where the  paths should be rerouted through
different, however ``equivalent'', flaps of the compass. To deal with this,  Robertson and Seymour defined  in
\cite{RobertsonS95b} the concept of wall {\sl homogeneity}. Roughly speaking, when a wall is homogenous
then the {\sl variety of the ways} that  paths may be routed through the flaps that are
inside some ``brick'' of the wall is the same for all  bricks. In \cite{RobertsonS95b} it was proved that
every big enough flat wall contains a still big homogeneous subwall where the claimed rerouting is possible,
with the help of later results of the Graph Minors series \cite{RobertsonS09XXI,RobertsonSGM22}.\smallskip

The definition of wall homogeneity in \cite{RobertsonS95b}  was based on the
 concept of
the {\sl vision of a flap} and was
quite particular to the problems it was dealing with. To our knowledge, after
\cite{RobertsonS95b}, no much use of homogeneity, as defined in \cite{RobertsonS95b},
was done for algorithmic purposes.
Most of the results where the irrelevant vertex technique
was applied concerned questions on  surface-embeddable
graphs where the wall is ``already'' disk-embedded and there is  no need  of homogeneity  (see e.g.~\cite{KaminskiN12find,MarxS07obta,KawarabayashiMR08asim,Mohar99alin,KaminskiT12cont,Kawarabayashi09plan,KobayashiK09algo,GolovachKPT09indu,KawarabayashiLR10reco,KawarabayashiKM10link}).
An indicative exception to this rule is the celebrated algorithm in \cite{GroheKMW11find,GroheKMW10find} for
the problem of checking whether $H$ is a topological minor of a graph $G$ where  some
notion of homogeneity, tailor-made for this problem, was introduced (see~\cite[Theorem 5.8]{GroheKMW10find} and
also~\cite{FominLPSZ20hitt}).\smallskip

In this paper we introduce an alternative notion of wall homogeneity that is simpler and more versatile
to use.  This is done in \autoref{label_establecidas} and is based on the framework introduced in {\bf  ($\beta$)}. Our definition may help  dealing with the wide variety of the problems
as it permits any version of finite index flap equivalency (for instance, flap equivalency based on {\sf MSOL}-expressibility). We accompany the definition with an {\sf FPT}-algorithm that finds a homogeneous
subwall. This, together with the main result of {\bf ($γ$)}, can  permit us
to find ``big-enough'' homogeneous walls with compasses of bounded treewidth. This, in turn,
will permit the answer of {\sf MSOL}-queries in parts of the compass and will allow more elaborated applications of the irrelevant vertex technique (such as those  used for problems on surface-embeddable graphs
in~\cite{GolovachKMT17thep,FominGT19modi}).

\vspace{-.35cm}

\paragraph{($ε$)  Regular flatness pairs and plane representations.}  We call a pair $(W,\frak{R})$ {\em flatness pair}
if $W$ is a flat wall  whose flatness is certified by the rendition $\frak{R}.$
Based on the framework of ($β$), in \autoref{label_definitionen} we  introduce a notion of {\sl regularity} for flatness pairs,
which roughly demands that  the branching vertices of the wall are ``internal''
with respect to the flaps of the compass of~$W.$ Regular flatness pairs
permit the representation
of  the compass of a flat wall by a graph embedded in a disk and a ``well-arranged''
wall inside it. This ``plane'' representation of flat walls will appear handy in other  applications.
For instance, it has been a useful tool  for the proofs of the main combinatorial results of~\cite{SauST20anfp,BasteST20acom}
as it makes it possible to translate routing questions inside compasses to
analogous questions on planar embeddings
and deal with them in a more easy way (using the new homogeneity concept of {\bf ($δ$)}).
%\ig{we should also mention our ICALP paper~\cite{} here, right?}

\subsection{Organization of the paper}
In \autoref{label_schlussgesetze} we provide some definitions and preliminary results and we state the two main results of this paper (\autoref{label_proporcionada} and \autoref{label_considerabil}), that assert the existence of an algorithm computing a tilt of a subwall of a flat wall and of an algorithm, that given a flatness pair outputs a regular flatness pair, respectively.
We prove \autoref{label_proporcionada} and \autoref{label_considerabil} in \autoref{dsanfldfalksdsa}.
In \autoref{dnajfndkjsfnsjk}, we develop the tools to address the topics $\textbf{(}\beta\textbf{)}$, $\textbf{(}\gamma\textbf{)}$, \textbf{(}$\delta$\textbf{)}, and \textbf{(}$\varepsilon$\textbf{)} listed above.

%\begin{itemize}
%	\item Subwalls of plane walls are plane walls while Subwalls of flat walls are NOT flat walls. How to deal with the lack of this virtual generalization. We answer to this issue from the formal point of view.
%
%	\item We find and fix an imprecision in the Flat Wall theorems statements (Kawarabayashi, Thomas, and Wollan\cite{KawarabayashiTW18}, Chuzhoy \cite{Chuzhoy15})
%	\item We explain how flatness of subwalls of flat walls should be considered (in every existing applications using flat subwalls): tilts of subwalls are flat not the subwalls themselves: cite papers where this is the case.
%	\item  We define a special type of levelings: a plane representation flat walls where cycles are corresponding to cycles maintaining the insulation properties. Insulation properties of the plane representation control the insulation properties of the flat wall.
%
%	      \sed{Perhaps we explain at some point interrelations of the two papers. However, this will come last.}
%
%	\item  We define an appropriate notion of flat wall such that we can correspond
%	      a flat wall to each of its subwalls with the same property. Also insulation properties
%	      of subwalls can be controlled by the plane representation of the initial graph.
%	\item  Based on our setting a  stronger and more versatile  version of homogeneity can be defined.
%	\item Application where compass is of bounded treewidth
%\end{itemize}

\section{Definitions and preliminary   results}\label{label_schlussgesetze}

\subsection{Preliminaries}
\label{label_descuartizada}

\paragraph{Sets and integers.}\label{label_suspendiendo}
We denote by $\Bbb{N}$ the set of non-negative integers.
Given two integers $p,q,$ where $p\leq q,$  we denote by $[p,q]$ the set $\{p,\ldots,q\}.$
For an integer $p\geq 1,$ we set $[p]=[1,p]$ and $\Bbb{N}_{\geq p}=\mathbb{N}\setminus [0,p-1].$
For a set $S,$ we denote by $2^{S}$ the set of all subsets of $S$ and by ${S \choose 2}$ the set of all subsets of $S$ of size $2.$
If ${\cal S}$ is a collection of objects where the operation $\cup$ is defined, then we denote $\cupall {\cal S}=\bigcup_{X\in {\cal S}}X.$

\paragraph{Basic concepts on graphs.}\label{label_produziertsein}
As a graph $G$ we denote any pair $(V,E)$ where $V$ is a finite set and $E\subseteq {V \choose 2},$ that is, all graphs of this paper are undirected, finite, and without loops or multiple edges.
We also define $V(G)=V$ and $E(G)=E.$
%We use standard graph-theoretic notation and we refer the reader to \cite{Die10} for any undefined terminology.
We say that a pair $(L,R)\in 2^{V(G)}\times 2^{V(G)}$ is a {\em separation} of $G$ if $L\cup R=V(G)$ and there is no edge in $G$ between $L\setminus R$ and $R\setminus L.$
Given a vertex $v\in V(G),$ we denote by $N_{G}(v)$ the set of vertices of $G$ that are adjacent to $v$ in $G.$
Also, given a set $S\subseteq V(G),$ we set $N_G(S) = \bigcup_{v \in S}N_G(v).$
% \ig{We need also this (used for instance in condition (iii) of tightness of a rendition): ``For a subgraph $H$ of $G$ and a set $S \subseteq V(G),$ we define $N_H(v) = N_G(v) \cap V(H)$ and $N_H(S) = \bigcup_{v \in S}N_H(v).$''}
A vertex $v \in V(G)$ is \emph{isolated} if $N_G(v) = \emptyset.$
For $S \subseteq V(G),$ we set $G[S]=(S,E\cap{S \choose 2} )$ and use  $G \setminus S$ to denote $G[V(G) \setminus S].$
Given an edge $e=\{u,v\}\in E(G),$ we define the {\em subdivision} of $e$ to be the operation of deleting $e,$ adding a new vertex $w,$ and making it adjacent to $u$ and $v.$
Given two graphs $H,G,$ we say that $H$ is a {\em subdivision} of $G$ if $H$ can be obtained from $G$ by subdividing edges.
The \emph{contraction} of an edge $e = \{u,v\}$ of a simple graph $G$ results in a simple graph $G'$ obtained from $G \setminus \{u,v\}$ by adding a new vertex $uv$ adjacent to all the vertices in the set $(N_G(u) \cup N_G(v))\setminus \{u,v\}.$
A graph $G'$ is a \emph{minor} of a graph $G$ if $G'$ can be obtained from a subgraph of $G$ after a series of edge contractions.

%\paragraph{Partially disk-embedded graphs.} %\label{app_pde}
%A {\em closed} (resp. {\em open}) {\em disk} is a set homeomorphic to the set $\{(x,y)\in \Bbb{R}^{2}\mid x^{2}+y^{2}\leq 1\}$ (resp. $\{(x,y)\in \Bbb{R}^{2}\mid x^{2}+y^{2}< 1\}$). Let $\Delta$ be a closed disk.
%We say that a graph $G$ is {\em $\Delta$-embedded} if $G$ is embedded in $\Delta$ without crossings such that the intersection of the boundary of $\Delta$ and $G$ (seen as a set of points of $\Delta$)  is  a subset of $V(G).$
%A {\em circle} of $\Delta$ is any set  homeomorphic to
%$\{(x,y)\in \Bbb{R}^{2}\mid x^{2}+y^{2}= 1\}.$
%  Given two distinct points $x,y\in D,$ an {\em $(x,y)$-arc} of $D$ is any subset of $D$ that is homeomorphic to the closed interval $[0,1].$\sed{We have to add {\em partially disk embedded}}.

\paragraph{Disk-embedded graphs.} %\label{app_pde}
A {\em closed} (resp. {\em open}) {\em disk} is a set homeomorphic to the set $\{(x,y)\in \Bbb{R}^{2}\mid x^{2}+y^{2}\leq 1\}$ (resp. $\{(x,y)\in \Bbb{R}^{2}\mid x^{2}+y^{2}< 1\}$).
Let $\Delta$ be a closed disk.
We use $\bd(\Delta)$ to denote the boundary of $\Delta$ and $\inter(\Delta)$ to denote the open disk $\Delta\setminus \bd(\Delta).$
When we embed a graph $G$ in the
plane or in a disk, we treat $G$ as a set of points. This permits us to make
set operations operations between graphs and sets of points.
We say that a graph $G$ is {\em $\Delta$-embedded} if $G$ is embedded in $\Delta$ without crossings such that the intersection of $\bd(\Delta)$ and $G$ (seen as a set of points of $\Delta$)  is a subset of $V(G).$

%We say that a graph $G$ is {\em partially disk-embedded in some closed disk $\Delta$},\sed{This is not necessary any more here!}
%if there is some subgraph $K$
%% \ig{since we ask later that $\bd(\Delta)$ is a cycle of $K,$ we don't need to require here that $K$ is non-empty, right?}
%of $G$ that is embedded in $\Delta$
%such that $\bd(\Delta)$ is a cycle of $K$  and $(V(G)\cap \Delta,V(G)\setminus\inter(\Delta))$
%is a separation of $G.$ From now on, we use the term {\em partially $\Delta$-embedded graph $G$}
%to denote that a graph $G$ is  partially disk-embedded in some closed disk $\Delta.$
%We also call the graph $K$
%{\em compass}
%of the partially $\Delta$-embedded graph $G$ and we always assume that we accompany
%a partially $\Delta$-embedded graph $G$ together with an embedding of its compass in $\Delta$ that is the set $G\cap \Delta.$

A {\em circle} of $\Delta$ is any set  homeomorphic to
$\{(x,y)\in \Bbb{R}^{2}\mid x^{2}+y^{2}= 1\}.$
Given two distinct points $x,y\in D,$ an {\em $(x,y)$-arc} of $D$ is any subset of $D$ that is homeomorphic to the closed interval $[0,1].$

%\subsubsection{pipes}

%
%
% If only edge contractions are allowed, we say that $G'$ is a \emph{contraction} of $G.$
%%
%Given two graphs $H,G,$ if $H$ is a minor of $G$ then for every vertex $v\in V(H)$ there is a set of vertices in $G$ that are the endpoints of the edges of $G$ contracted towards creating $v.$ We call this set {\em model} of $v$ in $G.$
%Given a finite collection of graphs ${\cal F}$ and a graph $G,$ we use notation ${\cal F}\leq_{\sf m} G$ to denote that some graph in ${\cal F}$ is a minor of $G.$
%
%We present here the main result of~\cite{BasteST20acom}.
%We will use this in order to solve {\sc ${\cal F}$-M-Deletion} on instances of bounded treewidth.

%\begin{proposition}
%%[Baste et al.~\cite{BasteST20acom} \ig{I think this is the only place where we add the name of the authors in a proposition, so we'd better remove it}]
%\label{label_rimproverando}
%	%\ign{${\sf size}({\cal F})$ has not been defined}
%	Let ${\cal F}$ be a finite collection of graphs and let $s_{\cal F}=\max\{|V(G)|\mid G\in {\cal F}\}.$
%	There exists an algorithm that given a triple $(G,\tw,k)$ where $G$ is a graph on $n$ vertices and of treewidth at most $\tw,$ and $k$ is a non-negative integer, it outputs, if it exists, a vertex set $S$ of $G$ of size at most $k$ such that ${\cal F}\nleq_{\sf m} G\setminus S.$
%	Moreover, this algorithm runs in  $2^{{\cal O}_{s_{\cal F}}(\tw \log \tw)}\cdot n$ time.
%\end{proposition}

\paragraph{Walls.}\label{label_dissunapiter}
Let  $k,r\in\Bbb{N}.$ The
\emph{$(k\times r)$-grid} is the
graph whose vertex set is $[k]\times[r]$ and two vertices $(i,j)$ and $(i',j')$ are adjacent if $|i-i'|+|j-j'|=1.$
An  \emph{elementary $r$-wall}, for some odd integer $r\geq 3,$ is the graph obtained from a
$(2 r\times r)$-grid
% \ig{grids are defined only in \autoref{label_monosyllable}. We should either move the definition here, or to cite appropriately where the definition can be found}
with vertices $(x,y)
	\in[2r]\times[r],$
after the removal of the
``vertical'' edges $\{(x,y),(x,y+1)\}$ for odd $x+y,$ and then the removal of
all vertices of degree one.
Notice that, as $r\geq 3,$  an elementary $r$-wall is a planar graph
that has a unique (up to topological isomorphism) embedding in the plane $\Bbb{R}^{2}$
such that all its finite faces are incident to exactly six
%\ig{we write ``six'' here, but ``2'' above, we have to unify. I prefer to use letters for integers up to ten}
edges.
The {\em perimeter} of an elementary $r$-wall is the cycle bounding its infinite face, while the cycles bounding its finite faces are called {\em bricks}.
Also, the vertices
in the perimeter of an elementary $r$-wall that have degree two are called {\em pegs},
while the vertices $(1,1), (2,r), (2r-1,1),$ and $(2r,r)$ are called {\em corners} (notice that the corners are also pegs).

\begin{figure}[h]
	\begin{center}
		\includegraphics[width=13cm]{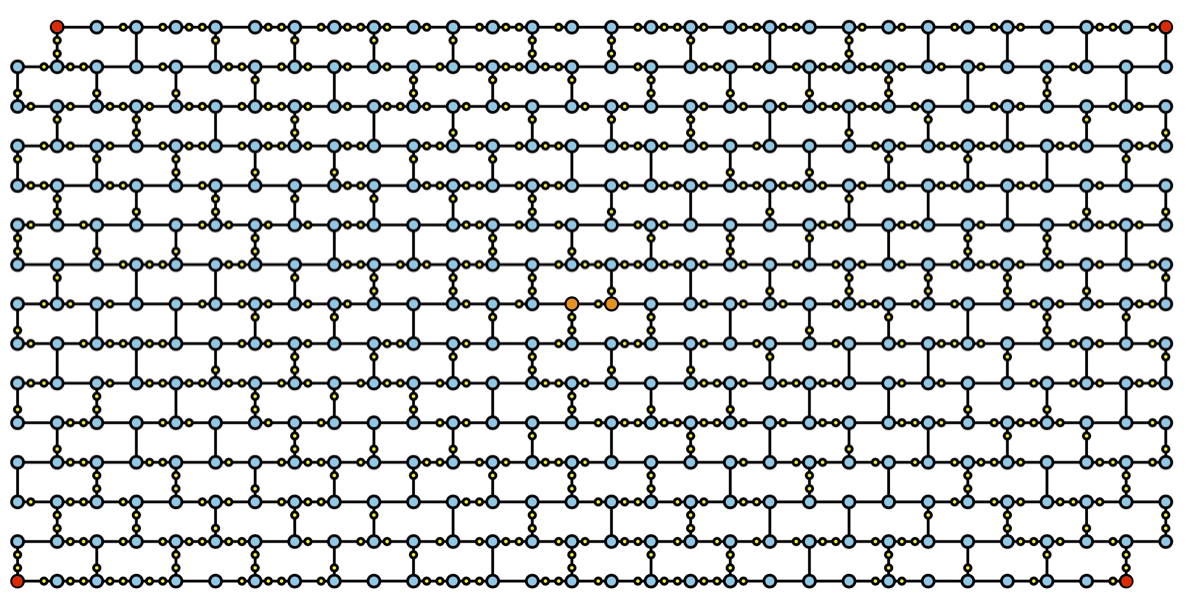}
	\end{center}
	\caption{A $15$-wall. The 3-branch vertices are depicted in cyan except from the corner and the central vertices that are depicted in red and orange respectively.}
	\label{label_thoroughfare}
\end{figure}

An {\em $r$-wall} is any graph $W$ obtained from an elementary $r$-wall $\bar{W}$
after subdividing edges (see \autoref{label_thoroughfare}). A graph $W$ is a {\em wall} if it is an $r$-wall for some odd $r\geq 3$
and we refer to $r$ as the {\em height} of $W.$ Given a graph $G,$
a {\em wall of} $G$ is a subgraph of $G$ that is a wall.
We
insist that, for every $r$-wall, the number $r$ is always odd.
%{: for this, whenever an $r$-wall appears with $r$ even, we agree to round it up to the next odd $r+1.$}

We call the vertices of degree three of a wall $W$ {\em 3-branch vertices}.
A cycle of $W$ is a {\em brick} (resp. the {\em perimeter}) of $W$ if its 3-branch vertices are the vertices of a brick (resp. the perimeter) of $\bar{W}.$
We denote by ${\cal C}(W)$ the set of all cycles of $W.$
We  use $D(W)$ in order to denote the perimeter of the  wall $W.$
%  we say that $W$ is a {\em wall} of $G.$
A brick of $W$ is {\em internal} if it is disjoint from $D(W).$

\paragraph{Subwalls.}

Given an elementary $r$-wall $\bar{W},$ some $i\in \{1,3,\ldots,2r-1\},$ and $i'=(i+1)/2,$
the {\em $i'$-th  vertical path} of $\bar{W}$  is the one whose
vertices, in order of appearance, are $(i,1),(i,2),(i+1,2),(i+1,3),
	(i,3),(i,4),(i+1,4),(i+1,5),
	(i,5),\ldots,(i,r-2),(i,r-1),(i+1,r-1),(i+1,r).$
Also, given some $j\in[2,r-1]$ the {\em $j$-th horizontal path} of $\bar{W}$
is the one whose
vertices, in order of appearance, are $(1,j),(2,j),\ldots,(2r,j).$

A \emph{vertical} (resp. \emph{horizontal}) path of $W$ is one
that is a subdivision of a  vertical (resp. horizontal) path of $\bar{W}.$
%whose 3-branch vertices are the vertices of a vertical (resp. horizontal) path of $\bar{W}.$
{Notice that the perimeter of an $r$-wall $W$ is uniquely defined regardless of the choice of the elementary $r$-wall $\bar{W}.$}
%\gstam{I think this observation is redundant. \ig{I think it is not completely redundant, I would keep it}}
A {\em subwall} of $W$ is any subgraph $W'$ of  $W$
that is an $r'$-wall, with $r' \leq r,$ and such the vertical (resp. horizontal) paths of $W'$ are subpaths of the
	{vertical} (resp. {horizontal}) paths of $W.$

%
%
%The {\em layers} of an $r$-wall $W$  are recursively defined as follows.
%The first layer of $W$ is its perimeter. For $i=2,\ldots,(r-1)/2,$ the $i$-th layer of $W$ is the $(i-1)$-th layer of the subwall $W'$ obtained from $W$ after removing from $W$ its perimeter and removing recursively all occurring vertices of degree one.
%We refer to the $(r-1)/2$-th layer as the {\em inner layer} of $W.$\sed{We need central. Add then in the figure!}
%The {\em central vertices} of an $r$-wall are its two 3-branch vertices  that do not belong to any of its layers.
%See \autoref{label_thoroughfare} for an illustration of the notions defined above.
%
%Given an $r$-wall $W$ and an odd $q\in\Bbb{N}_{\geq 3}$ where $r\geq q,$ we define the {\em central $q$-subwall} of $W,$ denoted by $W^{(q)},$ to be the $q$-wall obtained from $W$ after removing from $W$ its first $(r-q)/2$ layers and all occurring vertices of degree one.

%
%Given a  wall $W,$ we define its {\em inpegs}  as the vertices of its perimeter
%that are incident to edges of $W$  that are not in its perimeter.

\paragraph{Tilts.}
The {\em interior} of a wall $W$ is the graph obtained
from $W$ if we remove from it all edges of $D(W)$ and all vertices of $D(W)$ that have degree two in $W.$ Given two walls $W$ and $\tilde{W}$ of a graph $G,$ we say that $\tilde{W}$ is a {\em tilt} of $W$ if $\tilde{W}$ and $W$ have identical interiors.

\subsection{Renditions}

\paragraph{Paintings.}
Let $\Delta$ be a closed disk.
%, i.e., a set homeomorphic to the set $\{(x,y)\in \Bbb{R}^{2}\mid x^{2}+y^{2}\leq 1\}.$
Given a subset $X$ of $\Delta,$ we
denote its closure by $\bar{X}$ and its boundary by $\bd(X).$
A {\em {$\Delta$}-painting} is a pair $\Gamma=(U,N)$
%\rev{Please consider adding a picture explaining the concepts of paintings and cells. \ig{when we will include the definitions here, it will be more clear}}
where
\begin{itemize}
	\item  $N$ is a finite set of points of $\Delta,$
	\item $N \subseteq U \subseteq \Delta,$ and
	\item $U \setminus  N$ has finitely many arcwise-connected  components, called {\em cells}, where, for every cell $c,$
	      \begin{itemize}
		      \item[$\circ$] the closure $\bar{c}$ of $c$
		            is a closed disk
		            and
		      \item[$\circ$]  $|\tilde{c}|\leq 3,$ where $\tilde{c}:=\bd(c)\cap N.$
	      \end{itemize}
\end{itemize}
We use the  notation $U(\Gamma) := U,$
$N(\Gamma) := N$  and denote the set of cells of $\Gamma$
by $C(\Gamma).$
%Given a cell $c\in C(\Gamma)$  we will call the points in $\bd(c)\cap N$ {\em endpoints} of $c.$
For convenience, we may assume that each cell  of $\Gamma$ is an open disk of $\Delta.$

Notice that, given a $\Delta$-painting $\Gamma,$
the pair $(N(\Gamma),\{\tilde{c}\mid c\in C(\Gamma)\})$  is a hypergraph whose hyperedges have cardinality at most three and  $\Gamma$ can be seen as a plane embedding of this hypergraph in $\Delta.$

\paragraph{Renditions.} Let $G$ be a graph, and let $\Omega$ be a cyclic permutation of a subset of $V(G)$ that we denote by $V(\Omega).$ By an {\em $\Omega$-rendition} of $G$ we mean a triple $(\Gamma, \sigma, \pi),$ where
%\ig{I have put these three properties in itemize environment}
\begin{itemize}
	\item[(a)] $\Gamma$ is a $\Delta$-painting for some closed disk $\Delta,$
	\item[(b)] $\pi: N(\Gamma)\to V(G)$ is an injection, and
	\item[(c)] $\sigma$ assigns to each cell $c \in  C(\Gamma)$ a subgraph $\sigma(c)$ of $G,$ such that
	      \begin{enumerate}
		      \item[(1)] $G=\bigcup_{c\in C(\Gamma)}\sigma(c),$
		            %\gstam{Should we use default LIPIcs enumeration here?}
		      \item[(2)]  for distinct $c, c' \in  C(\Gamma),$  $\sigma(c)$ and $\sigma(c')$  are edge-disjoint,
		      \item[(3)] for every cell $c \in  C(\Gamma),$ $\pi(\tilde{c}) \subseteq V (\sigma(c)),$
		      \item[(4)]  for every cell $c \in  C(\Gamma),$  $V(\sigma(c)) \cap \bigcup_{c' \in  C(\Gamma) \setminus  \{c\}}V(\sigma(c')) \subseteq \pi(\tilde{c}),$ and
		      \item[(5)]  $\pi(N(\Gamma)\cap \bd(\Delta))=V(\Omega),$ such that the points in $N(\Gamma)\cap \bd(\Delta)$ appear in $\bd(\Delta)$ in the same ordering as their images, via $\pi,$ in $\Omega.$
	      \end{enumerate}
\end{itemize}
Given an $\Omega$-rendition $(\Gamma, \sigma, \pi)$ of a graph $G,$ we call a cell $c$ of $\Gamma$ {\em trivial} if $\pi(\tilde{c})=V(\sigma(c)).$

We say that an  {$\Omega$-rendition}  $(\Gamma, \sigma, \pi)$ of a graph $G$ is {\em tight} if the following conditions are satisfied:

\begin{enumerate}

	\item[(i)] If there are two points $x,y$ of $N(\Gamma)$
	      such that $e=\{\pi(x),\pi(y)\}\in E(G),$ then
	      there is a cell $c\in C(\Gamma)$ such that $\sigma(c)$ is
	      the  two-vertex connected graph $(e,\{e\}),$

	\item[(ii)]  for every $c\in C(\Gamma),$ every two vertices in $\pi(\tilde{c})$ belong to some path of $\sigma(c),$

	\item[(iii)] for every $c \in  C(\Gamma)$ and every connected component $C$ of the graph
	      $\sigma(c)\setminus \pi(\tilde{c}),$  if $N_{\sigma(c)}(V(C))\neq\emptyset,$ then $N_{\sigma(c)}(V(C))=\pi(\tilde{c}),$

	\item[(iv)] there are no two distinct non-trivial cells $c_{1}$ and $c_{2}$ such that  $\pi(\tilde{c_1})=\pi(\tilde{c_2}),$ and%\sed{Check the last two in relation to SODA 2020}

	\item[(v)] 	for every $c \in  C(\Gamma)$ there are
	      $|\tilde{c}|$ vertex-disjoint paths in $G$ from $\pi(\tilde{c})$ to the set $V(\Omega).$
\end{enumerate}

\begin{figure}[h]
	\begin{center}
		\includegraphics[width=13cm]{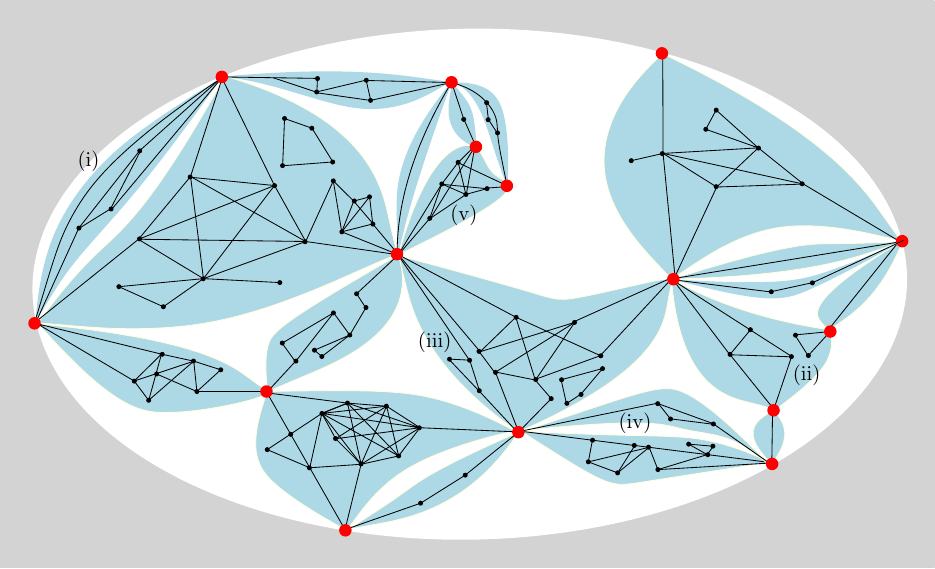}
	\end{center}
	\caption{A graph $G$ together with an  $\Omega$-rendition of $G,$  where all tightness conditions are violated.}
	\label{label_operationalization}
\end{figure}

%
% The proof of the next lemma is postponed to  \autoref{label_subssssslevaciones}.

\begin{lemma}
	\label{label_levantadores}
	There is a linear-time algorithm  that, given a graph  $G$ and an $\Omega$-rendition $(\Gamma, \sigma, \pi)$ of $G,$ outputs a {tight} $\Omega$-rendition of $G.$
\end{lemma}

\begin{proof}%[Proof of \autoref{label_levantadores}.]
	We argue about how to transform   $(\Gamma, \sigma, \pi)$  to   a tight  {$\Omega$-rendition}  of $G$ in ${\cal O}(n+m)$ time.
	See \autoref{label_operationalization} for an example of a graph $G$ together with an $\Omega$-rendition of $G$ that {\sl violates} each of the five tightness conditions (indicated in the figure).

	For the first property, let $e=\{\pi(x),\pi(y)\}\in E(G)$ be an edge of $G$ that belongs to some $\sigma(c)$ with $|V(\sigma(c))|>2.$ Then, we add a new cell $c_{\rm new}$ to the rendition,
	where $\pi(\tilde{c}_{\rm new})=\{\pi(x),\pi(y)\}$ and $\sigma(c_{\rm new})=(e,\{e\})$
	Also, we remove the edge $e$ from $\sigma(c).$

	For the second property, let $c$ be a cell in $C(\Gamma)$
	%\ig{we should say: ``let $c$ be a cell and ...''}
	and let ${\cal C}$ be a collection containing every component  of the graph $\sigma(c).$
	We say that $C_{1},C_{2}\in {\cal C}$ are {\em equivalent} if $V(C_{1})\cap \pi(\tilde{c})=V(C_{2})\cap \pi(\tilde{c}).$ Notice that each equivalence class of this equivalence relation corresponds to some
	partition ${\cal P}$ of $\pi(\tilde{c}).$  If this equivalence relation has only one class, then (ii) holds,
	% \ig{we may say that this is because of condition (c.3) of the definition of rendition}
	because of condition (c.3) of the definition of rendition.
	If not, we remove $c$ from the rendition and we replace it with as many cells as the number of equivalence classes,
	% \ig{it is not $|{\cal P}|,$ but the number of equivalence classes. By the way, maybe we want to give a name to this equivalence relation?}
	one for each equivalence class and we update $\sigma$ so that each new cell is mapped
	to the union of the members of the equivalence class
	% \ig{it should be ``equivalence class'' (in singular)}
	corresponding to it.

	For the third property, consider some $c\in C(\Gamma),$ and observe that, because of (i) and (ii),
	the graph	$\sigma(c)\setminus \pi(\tilde{c})$ contains at least one connected component, say $C^*,$
	with  $N_{\sigma(c)}(V(C))\neq \emptyset.$
	Let ${\cal C}$ be a collection containing every component  of the graph $\sigma(c)\setminus \pi(\tilde{c}).$
	We say that $C_{1},C_{2}\in {\cal C}$ are {\em equivalent} if $N_{\sigma(c)}(V(C_{1}))=N_{\sigma(c)}(V(C_{2})).$ Notice that this equivalence relation has at most eight equivalence classes, each corresponding to a subset of $\pi(\tilde{c}).$ For each subset $X$ of $\pi(\tilde{c}),$ we define the graph $F_{X}$ as the union of the graphs in the corresponding equivalence class.
	Let also $X^*$  be the non-empty subset of $\pi(\tilde{c})$ such that $C^*$ is a subgraph of $F_{X^*}.$
	We enhance $F_{X^*}:=F_{X^*}\cup  F_{\emptyset}.$
	We now remove the cell $c$ from the rendition
	and  for every non-empty $X\in 2^{\pi(\tilde{c})}$ where $F_{X}$ is non-null, we add a new cell $c_{X}$ and we update $\sigma$ by mapping each $c_{X}$ to the graph $F_{X}.$

	For property (iv), for every two distinct non-trivial cells $c_{1}$ and $c_{2}$ with $\pi(\tilde{c_1})=\pi(\tilde{c_2}),$ we remove $c_{2}$ from the rendition and we update $\sigma(c_{1}):=\sigma(c_{1})\cup \sigma(c_{2}).$

	The last property
	can be achieved as follows: we first
	construct an auxiliary planar graph $G'$ by substituting in $G$ each $\sigma(c)$ by a clique on $\pi(\tilde{c})$ (that is a vertex, an  edge, or a triangle) and by adding a new vertex $v_{\rm new}$ adjacent to all the vertices in $V(\Omega)$; then the new rendition can be easily constructed  starting from  the triconnected component  $C$ % \sed{NEW!}
	of $G'$ that contains $v_{\rm new}$ (to find the triconnected components, one may use the classic algorithm of Hopcroft and Tarjan\cite{HopcroftT73divi} that runs in ${\cal O}(n+m)$ time)  and then attaching to $C,$ as images of the updated $\pi,$
	the other triconnected components.
	%
	%
	%\ig{I think we should explain in more detail how we construct the new rendition}.
\end{proof}

In the rest of this paper we use only conditions (i)--(iii) of the tightness definition. However, we adopt the above, more strict, version of tightness as it will be useful in further applications.

\subsection{Flatness pairs}
\label{label_prognostication}
Let $W$ be an $r$-wall, for some odd integer $r\geq 3.$ We say that a pair $(P,C)\subseteq D(W)\times D(W)$ is a {\em choice
of pegs and corners for $W$} if $W$ is the subdivision of an  elementary $r$-wall $\bar{W}$
where $P$ and
$C$ are the pegs and the corners of $\bar{W},$ respectively (clearly, $C\subseteq P$).
%{\gstam{I added this sentence here but maybe it is unnecessary. \ig{such sentences are really helpful}}
To get more intuition, notice that a wall $W$ can occur in several ways from the elementary wall $\bar{W},$ depending on the way the vertices in the perimeter of $\bar{W}$ are subdivided. Each of them %\sed{Height of a flatness pair}
gives a different selection $(P,C)$ of pegs and corners of $W.$

Let an odd integer $r\geq 3$ and $W$ be an $r$-wall of some graph $G.$
We say that $W$ is a {\em flat $r$-wall}
of $G$ if there is a separation $(X,Y)$ of $G$ and a choice  $(P,C)$
of pegs and corners for $W$ such that:
\begin{itemize}
	\item $V(W)\subseteq Y,$
	\item  $P\subseteq X\cap Y\subseteq V(D(W)),$ and
	\item  if $\Omega$ is the cyclic ordering of the vertices $X\cap Y$ as they appear in $D(W),$ then there exists an $\Omega$-rendition $(\Gamma,\sigma,\pi)$ of  $G[Y].$
\end{itemize}

Because of \autoref{label_levantadores}, we can assume (and we also demand) that the  $\Omega$-rendition $(\Gamma,\sigma,\pi)$ of  $G[Y]$ in the above definition is always tight. We mention here that Chuzhoy~\cite{Chuzhoy15impr} uses a slightly different notion of flatness, where the separation $(X,Y)$ consists of two edge-disjoint {\sl subgraphs}, instead of two {\sl vertex sets}, and where the graph $Y$ may play the role of the compass.

\paragraph{Flatness pairs.}
Given the above, we  say that  the choice of the 7-tuple $\frak{R}=(X,Y,P,C,\Gamma,\sigma,\pi)$  {\em certifies 	that $W$ is a flat wall of $G$}. We call the pair $(W,\frak{R})$ a {\em flatness pair} of $G$ and define the {\em height} of the pair $(W,\frak{R})$ to be the height of $W.$
%	{If $(\Gamma,\sigma,\pi)$ is tight then we also say that $\frak{R}$ (and also $(W,\frak{R})$) is {\em tight}.}\gstam{This extra definition is not used anywhere.}
We use the term {\em cell of} $\frak{R}$ in order to refer to the cells of $\Gamma.$

We call the graph $G[Y]$ the {\em $\frak{R}$-compass} of $W$ in $G,$ denoted by ${\sf compass}_{\frak{R}}(W).$ We define the  {\em flaps} of the wall $W$ in $\frak{R}$ as ${\sf flaps}_{\frak{R}}(W):=\{\sigma(c)\mid c\in C(\Gamma)\}.$  Given a flap $F\in {\sf flaps}_{\frak{R}}(W),$ we define its {\em base} as $\partial F:=V(F)\cap \pi(N(\Gamma)).$
A flap $F\in {\sf flaps}_{\frak{R}}(W)$ is {\em trivial} if  $|\partial F|=2$ and $F$ consists of one edge between the two vertices in $\partial F.$
We call the edges of the trivial flaps {\em short edges of ${\sf compass}_{\frak{R}}(W)$}. A  cell $c$ of ${\frR}$ is {\em untidy} if  $\pi(\tilde{c})$ contains a vertex
$x$ of ${W}$ such that two of the edges of ${W}$ that are incident to $x$ are edges of $\sigma(c).$ Notice that if $c$ is untidy then  $|\tilde{c}|=3.$

\begin{figure}[h]
	\begin{center}
		\reflectbox{\includegraphics[width=15cm]{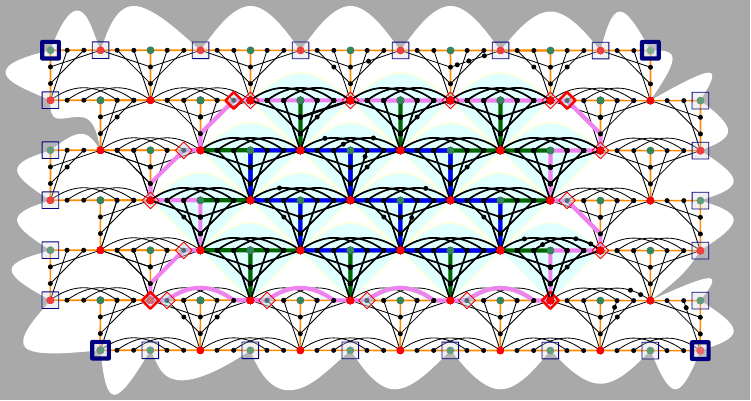}}
	\end{center}
	\caption{
		A flat 7-wall $W$ in a graph $G$ whose flatness is
		certified by some rendition $\frak{R}$ where the choice of pegs and corners in $\frak{R}$ corresponds to the squared  vertices. We depict only the $\frak{R}$-compass of $W$
		that consists of $W$ and some ``black paths'' between the vertices of $W.$
		The  5-wall $\tilde{W}'$ consisting of the fat edges (purple, green, blue) is a flat $\frak{R}$-normal wall
		of ${\sf compass}_{\frak{R}}(W).$ The flatness of $\tilde{W}'$ is  certified by the rendition $\tilde{\frak{R}}'=(X',Y',P',C',\Gamma',\sigma',\pi'),$ where $X'$ contains all the vertices incident to at least one orange edge plus the non-depicted vertices in the grey area, $Y'$ contains all vertices that are either in a ``fat'' black path or incident to at least two fat edges, the pegs are the diamond vertices, and the corners are the fat diamond vertices (that are also pegs). For the (tight)  $\Omega'$-rendition $(\Gamma',\sigma',\pi')$ of $G[Y'],$ see  \autoref{label_wahrheitsargu}.
	}
	\label{label_consumadamente}
\end{figure}

In \autoref{label_consumadamente} we depict a flat wall $W$ in a graph $G$ as well as the $\frak{R}$-compass of $W$ in $G,$  for some rendition $\frak{R}$ certifying its flatness.
Notice that there is a unique subwall $W'$ of $W$
that is disjoint from $D(W)$ and has height five.
Interestingly, the subwall $W'$  is {\sl not} a flat wall of $G,$ however
there is a tilt $\tilde{W}'$ of $W'$ that is a flat wall of $G.$
The wall $\tilde{W}'$ is depicted in \autoref{label_consumadamente}  and the
rendition certifying its flatness is depicted in \autoref{label_wahrheitsargu}.

\begin{figure}[h]
	\begin{center}
		\reflectbox{\includegraphics[width=12cm]{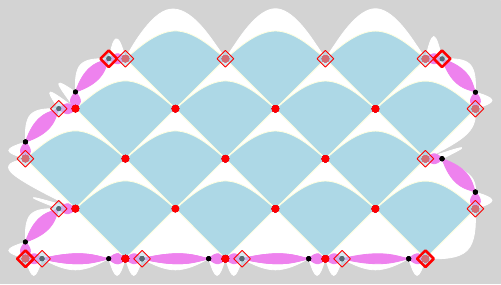}}
	\end{center}
	\caption{The painting of the rendition $\tilde{\frak{R}}'$ certifying the flatness of the  5-wall $\tilde{W}'$ of \autoref{label_consumadamente}. The  $\tilde{\frak{R}}'$-compass of $\tilde{W}'$ has two types of flaps: those whose base has  three vertices (they are images of the blue cells) and those that are trivial (they are images of the purple cells). }
	\label{label_wahrheitsargu}
\end{figure}

\begin{figure}[ht]
	\begin{center}
		\includegraphics[width=14cm]{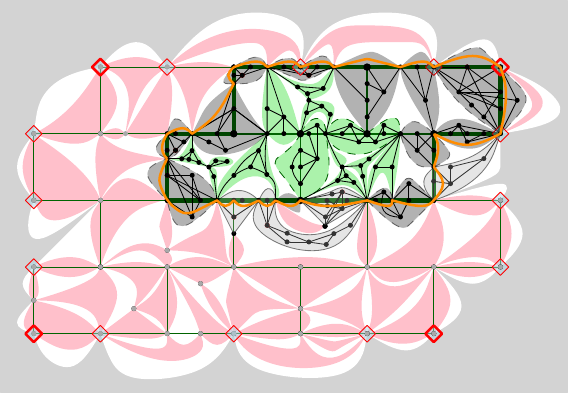}
	\end{center}		\vspace{-2mm}
	\caption{
		A flat wall $W$ in a graph $G,$ the painting of a rendition $\frak{R}$ certifying its flatness, a subwall $W'$ of $W,$ of height three, which is
		$\frak{R}$-normal, and the $\frak{R}$-flaps of $W,$ that correspond to  either $W'$-perimetric (depicted in grey) or $W'$-internal cells (depicted in green).
		The circle $K_{W'}$ is the fat orange cycle. The $W'$-marginal cells are depicted in light grey and the untidy cells are those with dashed boundary.
	}
	%
	%	{An $\Omega$-rendition of $G[Y]$ for some flatness pair $({W},{\frak{R}})$ of a graph $G$ and a subwall $W'$ of $W$ (that is $\frak{R}$-normal).
	%	The edges and the non-boundary vertices of the flaps corresponding to the $D$-strictly external cells of ${\frak{R}}$ (depicted in pink) are not depicted (however their boundary vertices that are not in $D$ are depicted in grey).
	%	There are nine $D$-outer-perimetric cells  of ${\frak{R}}$ (in blue) and seven $D$-inner-perimetric cells (in yellow). Also, there are five $D$-internal cells of ${\frak{R}}$ (in green). Among the $D$-internal cells of ${\frak{R}},$ those that are untidy are depicted with a dashed boundary. The orange cycle is the boundary of the disk $\Delta_{W'}.$}}
	\label{label_exhalaciones}
\end{figure}

\paragraph{Cell classification.}
Given a cycle $C$ of ${\sf compass}_{\frak{R}}(W),$ we say that
$C$ is {\em $\frak{R}$-normal} if it is not a subgraph of a flap $F\in {\sf flaps}_{\frak{R}}(W).$
Given an $\frak{R}$-normal cycle $C$ of ${\sf compass}_{\frak{R}}(W),$
we call a cell $c$ of $\frak{R}$ {\em $C$-perimetric} if   $\sigma(c)$ contains some edge of $C.$ Notice that if $c$ is $C$-perimetric, then $\pi(\tilde{c})$ contains two points $p,q\in N(\Gamma)$
such that  $\pi(p)$ and $\pi(q)$ are vertices of $C$ where one, say $P_{c}^{\rm in},$ of the two $(\pi(p),\pi(q))$-subpaths of $C$ is a subgraph of $\sigma(c)$ and the other, denoted by $P_{c}^{\rm out},$  $(\pi(p),\pi(q))$-subpath contains at most one internal vertex of $\sigma(c),$ which should be the (unique) vertex $z$ in $\partial\sigma(c)\setminus\{\pi(p),\pi(q)\}.$
We pick a $(p,q)$-arc $A_{c}$ in $\hat{c}:={c}\cup\tilde{c}$ such that  $\pi^{-1}(z)\in A_{c}$ if and only if $P_{c}^{\rm in}$ contains
the vertex $z$ as an internal vertex.

We consider the circle  $K_{C}=\cupall\{A_{c}\mid \mbox{$c$ is a $C$-perimetric cell of $\frak{R}$}\}$
and we denote by $\Delta_{C}$ the closed disk bounded by $K_{C}$  that is contained in  $\Delta.$
A cell $c$ of $\frak{R}$ is called {\em $C$-internal} if $c\subseteq \Delta_{C}$
and is called {\em $C$-external} if $\Delta_{C}\cap c=\emptyset.$
Notice that  the cells of $\frak{R}$ are partitioned into  $C$-internal,  $C$-perimetric, and  $C$-external cells.

Let $c$ be a tidy $C$-perimetric cell of $\frak{R}$ where $|\tilde{c}|=3.$ Notice that $c\setminus A_{c}$ has two arcwise-connected components and one of them is an open disk $D_{c}$ that is a subset of $\Delta_{C}.$
If the closure $\overline{D}_{c}$  of $D_{c}$ contains only two points of $\tilde{c}$ then we call the cell $c$ {\em $C$-marginal}.

%A cell of $\frak{R}$ is called {\em internal/marginal/external} if it is   $W$-internal/marginal/external.
\paragraph{Influence.}
For every $\frak{R}$-normal cycle $C$ of ${\sf compass}_{\frak{R}}(W)$ we define the set
$${\sf influence}_{\frak{R}}(C)=\{\sigma(c)\mid \mbox{$c$ is a cell of $\frak{R}$ that is not $C$-external}\}.$$
%
%
%graph $${\sf influence}_{\frak{R}}(W')=\cupall\{\pi(c)\mid \mbox{$c$ is a cell of $\frak{R}$ that is not $W'$-external}\}.$$

A wall $W'$  of ${\sf compass}_{\frak{R}}(W)$  is $\frak{R}$-normal if $D(W')$ is {\em $\frak{R}$-normal}.
Notice that every wall of $W$ (and hence every subwall of $W$) is an $\frak{R}$-normal wall of ${\sf compass}_{\frak{R}}(W).$ We denote by ${\cal S}_{\frak{R}}(W)$ the set of all $\frak{R}$-normal walls of ${\sf compass}_{\frak{R}}(W).$ Given a $W'\in {\cal S}_{\frak{R}}(W)$ and a cell $c$ of $\frak{R}$
we say that $c$ is {\em $W'$-perimetric/internal/external/marginal} if $c$ is  $D(W')$-perimetric/internal/external/marginal.
% (see  \autoref{label_impersonates} for an example).
We also use $K_{W'},$ $\Delta_{W'},$ ${\sf influence}_{\frak{R}}(W')$ as shortcuts
for $K_{D(W')},$ $\Delta_{D(W')},$ ${\sf influence}_{\frak{R}}(D(W')).$

\paragraph{Regular pairs.} Let $(W,\frak{R})$ be a flatness pair of a graph $G.$
We call a  flatness pair $(W,\frak{R})$ of a graph $G$ {\em regular}
if none of its cells is $W$-external, $W$-marginal, or untidy.

\paragraph{Tilts of flatness pairs.} Let $(W,\frak{R})$ and $(\tilde{W}',\tilde{\frak{R}}')$  be two flatness pairs of a graph $G$ and let $W'\in {\cal S}_{\frak{R}}(W).$ We also assume that ${\frak{R}}=(X,Y,P,C,\Gamma,\sigma,\pi)$ and $\tilde{\frak{R}}'=(X',Y',P',C',\Gamma',\sigma',\pi').$
We say that   $(\tilde{W}',\tilde{\frak{R}}')$   is a {\em $W'$-tilt}
of $(W,\frak{R})$ if \begin{itemize}
	\item $\tilde{\frak{R}}'$ does not have $\tilde{W}'$-external cells,
	\item  $\tilde{W}'$ is a tilt of $W',$
	\item  the set of $\tilde{W}'$-internal  cells of  $\tilde{\frak{R}}'$ is the same as the set of $W'$-internal cells of ${\frak{R}}$ and their images via $\sigma'$ and ${\sigma}$ are also the same,
	\item ${\sf compass}_{\tilde{\frak{R}}'}(\tilde{W}')$ is a subgraph of $\cupall{\sf influence}_{{\frak{R}}}(W'),$ and
\item if $c$ is a cell in $C(\Gamma') \setminus C(\Gamma),$ then $|\tilde{c}| \leq 2.$
\end{itemize}

The next observation follows from the definitions of regular flatness pairs and tilts.
\begin{observation}\label{label_riconoscendo}
If $(W,\frak{R})$ is a regular flatness pair, then for every $W'\in {\cal S}_{\frak{R}}(W)$ every $W'$-tilt of $(W,\frak{R})$ is also regular.
\end{observation}

\smallskip\smallskip\smallskip

The main results of this paper are the following.

\begin{theorem}
	\label{label_proporcionada}
	There exists an algorithm that given a graph $G,$ a flatness pair $({W},{\frak{R}})$ of $G,$ and a wall $W'\in {\cal S}_{\frak{R}}(W),$ outputs  a  $W'$-tilt of $({W},{\frak{R}})$ in  ${\cal O}(n+m)$ time.
\end{theorem}

\begin{theorem}
	\label{label_considerabil}
	There is an algorithm that, given a graph $G$
	and a flatness pair  $({W},{\frak{R}})$ of $G,$ outputs a  regular flatness pair $({W}^{\star},{\frak{R}}^{\star})$ of $G,$ with the same height as $({W},{\frak{R}})$  such that ${\sf compass}_{\frak{R}^{\star}}(W^{\star})\subseteq {\sf compass}_{\frak{R}}(W).$ This algorithm runs  in ${\cal O}(n+m)$ time.
	% admitting a ${\cal O}(n+m)$-computable regular tilt-assignment function $\theta$ {}.
\end{theorem}

\section{Applications}\label{dnajfndkjsfnsjk}

In this section we apply \autoref{label_proporcionada} and \autoref{label_considerabil} in order to address the items $\textbf{(}\beta\textbf{)}$, $\textbf{(}\gamma\textbf{)}$, \textbf{(}$\delta$\textbf{)}, and~\textbf{(}$\varepsilon$\textbf{)} discussed in the introduction.

\subsection{Tilts of subwalls}

We  present the following result from \cite{KawarabayashiTW18anew}, stated in our new framework.

\begin{proposition}\label{label_inconsiderable}
	There are two functions  $\newfun{label_questionnaires}:\Bbb{N}\to \Bbb{N}$  and
	$\newfun{label_hierarchical}:\Bbb{N}\to \Bbb{N}$  and
	an algorithm that receives as  input  a graph $G,$ an odd integer $r\geq 3,$ a $t\in\Bbb{N}_{\geq 1},$ and  an $\funref{label_questionnaires}(t)\cdot r$-wall  ${W}$ in $G,$ and outputs, in \change{${\cal O}(t^{24} \cdot m + n)$} time,
	\begin{itemize}
		\item  either that $K_{t}$ is a minor of $G$ or
		\item a set $A\subseteq V(G)$ where $|A|\leq \funref{label_hierarchical}(t)$  and a flatness pair $(\tilde{W}',\tilde{\frak{R}}')$ of $G\setminus A$ of height $r,$ such that $\tilde{W}'$ is a tilt of a subwall $W'$ of $W.$
	\end{itemize}
	Moreover $\funref{label_questionnaires}(t)=\Ocal(t^{26})$ and $\funref{label_hierarchical}(t)=\Ocal(t^{24}).$
\end{proposition}
An alternative of the above where  $\funref{label_questionnaires}(t)=\Ocal(t^2)$ and $\funref{label_hierarchical}(t)=t-5=\Ocal(t)$ has been proved by Chuzhoy in \cite{Chuzhoy15impr}  with a running time
that is polynomial in the input size.
However, we prefer the version of Kawarabayashi, Thomas, and Wollan \cite{KawarabayashiTW18anew} as their algorithm is linear.

%\sed{Comment on constants}

\subsection{Apex-walls with compasses of bounded treewidth }
\label{label_distrettamente}

We first define the notion of treewidth. A \emph{tree decomposition} of a graph~$G$
is a pair~$(T,\chi)$ where $T$ is a tree and $\chi: V(T)\to 2^{V(G)}$
such that
\begin{enumerate}
	\item $\bigcup_{t \in V(T)} \chi(t) = V(G),$
	\item for every edge~$e$ of~$G$ there is a $t\in V(T)$ such that
	      $\chi(t)$
	      contains both endpoints of~$e,$ and
	\item for every~$v \in V(G),$ the subgraph of~${T}$
	      induced by $\{t \in V(T)\mid {v \in \chi(t)}\}$ is connected.
\end{enumerate}
The \emph{width} of $(T,\chi)$ is defined as
$\w(T,\chi):=
	\max\big\{\left|\chi(t)\right|-1 \bigmid t\in V(T)\big\}.$
The \emph{treewidth of $G$} is defined as
$$\tw(G):=\min\big\{\w(T,\chi) \bigmid (T,\chi) \text{ is a tree decomposition of }G\big\}.$$

This subsection is dedicated to the proof of the following result.

\begin{theorem}\label{label_proletarians}
There is a function   $\newfun{label_confrontation}:\Bbb{N}\to \Bbb{N}$
and
an algorithm that receives as
input a graph $G,$ an odd integer $r\geq 3,$ and a
$t\in\Bbb{N}_{\geq 1},$ and outputs,
in
$2^{{\cal O}_t (r^2)}\cdot n$ time, one of the following:
\begin{itemize}
\item a report  that $K_{t}$ is a minor of $G,$
\item a tree decomposition of $G$ of width at most $\funref{label_confrontation}(t)\cdot r,$ or
\item a set $A\subseteq V(G)$,  where $|A|\leq \funref{label_hierarchical}(t),$ a regular flatness pair $(W,\frak{R})$ of $G\setminus A$ of height $r,$
and a tree decomposition of the $\frak{R}$-compass of $W$ of width at most $\funref{label_confrontation}(t)\cdot r.$  (Here $\funref{label_hierarchical}(t)$ is the function of \autoref{label_inconsiderable} and $\funref{label_confrontation}(t)=2^{\Ocal(t^2 \log t)}.$)
\end{itemize}
\change{Moreover, to obtain an explicit dependence on $t$, this algorithm can be modified to run in time $2^{2^{{\cal O}(t^2\log t)} r\log r + {\cal O}(r^2)}\cdot n+2^{2^{{\cal O}(t^2\log t)} r^3\log r}$.}
\end{theorem}

We will need some additional results in order to prove \autoref{label_proletarians}.
First we need the following result that is derived from \cite{PerkovicR00anim}. For a detailed analysis of the results of \cite{PerkovicR00anim}, see \cite{AlthausZ19opti}.

\begin{proposition}\label{label_verneinenden}
	There exists an algorithm with the following specifications:\medskip

	\noindent{\textbf{Input}:}	A graph $G$ and a non-negative integer $k$ such that $|V(G)|\geq 12k^{3}.$\\
	\noindent{\textbf{Output}:} A graph $G^{*}$ such that $|V(G^{*})|\leq (1-\frac{1}{16k^{2}}) \cdot |V(G)|$ and:
	\begin{itemize}
		\item Either $G^{*}$ is a subgraph of $G$ such that $\tw(G)=\tw({G^*}),$  or
		\item $G^{*}$ is obtained from $G$ after identifying the vertices of a matching in $G.$
	\end{itemize}
	Moreover, this algorithm runs in $2^{{\cal O}(k)} \cdot n$ time.
\end{proposition}

The following result of Kawarabayashi  and  Kobayashi \cite{KawarabayashiK20line},
provides a {\sl linear} relation between the treewidth and the height of a largest wall in a minor-free graph.

\begin{proposition}\label{label_scommettendo}
	There is a function $\newfun{label_entstandenen}:\Bbb{N}\to \Bbb{N}$ such that, for every $t,r\in \Bbb{N}$ and every graph $G$ that does not contain $K_{t}$ as a minor, if $\tw(G)\geq \funref{label_entstandenen}(t)\cdot r,$ then $G$ contains an $r$-wall. %\gstam{I would prefer to change ``$h$'' to ``$t$'' here.}
	In particular, one may choose $\funref{label_entstandenen}(t)=2^{{\cal O}(t^{2}\log t)}.$
\end{proposition}

The following is the main result of \cite{BodlaenderDDFLP16ackn}. We will use it to compute a tree decomposition of a graph of bounded treewidth.
%\gstam{These two results are used both in the proof of \autoref{label_hypostatizing} and \autoref{label_disproportionately}. Maybe we can place them when we introduce treewidth. \ig{it seems a good idea}}
\begin{proposition}\label{label_naturalistes}
	There is an algorithm that, given a graph $G$  and an integer $k,$ outputs either a report that $\tw(G)>k,$ or a tree decomposition of $G$ of width at most $5k+4.$
	Moreover, this algorithm runs in $2^{{\cal O}(k)} \cdot n$ time.
\end{proposition}

The following result is derived from \cite{AdlerDFST11fast}. We will use it in order to find a wall in a graph of bounded treewidth, given a tree decomposition of it.
\begin{proposition}\label{label_prohibitivas}
	There is an algorithm that, given a graph $G,$ a graph $H$ on $h$ edges  without isolated vertices, and a tree decomposition of $G$ of width at most $k,$ outputs, if it exists, a minor of $G$ isomorphic to $H.$
	Moreover, this algorithm runs in $2^{{\cal O}(k\log k)}\cdot h^{{\cal O}(k)}\cdot 2^{{\cal O}(h)}\cdot m$ time.
\end{proposition}

We start by proving the following ``light version'' of \autoref{label_proletarians}.

\begin{lemma}\label{label_hypostatizing}
	There exists an algorithm as follows:\\

	\noindent{\tt Find-Wall}$(G,t,r)$\\
	\noindent{\textbf{Input}:} A graph $G,$ an odd $r\in\Bbb{N}_{\geq 3},$  and a $t\in \Bbb{N}_{\geq 1}.$\\
	\noindent{\textbf{Output}:} One of the following:
	\begin{itemize}
		\item a report  that $K_{t}$ is a minor of $G,$
		\item a report that $G$ has treewidth at most $\funref{label_entstandenen}(t)\cdot r,$ where  $\funref{label_entstandenen}$ is as in \autoref{label_scommettendo}, or

		\item an $r$-wall $W$ of $G.$
		      %$ is a \no-instance of \mnb{\sc ${\cal F}$-M-Deletion}.
	\end{itemize}
	Moreover, this algorithm runs in  $2^{{\cal O}_{t}(r^2)}\cdot n$ time.
\change{To obtain an explicit dependence on $t$, this algorithm can be modified to run in time $2^{2^{{\cal O}(t^2\log t)} r\log r + {\cal O}(r^2)}\cdot n+2^{2^{{\cal O}(t^2\log t)} r^3\log r}$.}
\end{lemma}

%
%\autoref{label_hypostatizing} is restated in the Appendix
%in a slightly more  general version as \autoref{algoA}.

%{The proof of \autoref{label_hypostatizing} combines  the algorithm of Perkovi\v{c} and Reed in~\cite{PerkovicR20anim} for computing the treewidth of a graph, as well as the  excellent analysis of the algorithm provided in~\cite{AlthausZ19}.\gstam{This is a sketch of the proof.}
%We also use the upper bound for the treewidth of an $K_{h}$-minor free graph without an $r$-wall by~\cite{KawarabayashiK12a},
%the dynamic programming algorithm of~\cite{AdlerDFST11}
%%\includegraphics[width=3.4cm]{algdp}.
%for finding a wall in a graph of bounded treewidth, and the
%single-exponential {\sf FPT}-approximation algorithm for treewidth in~\cite{BodlaenderDDFLP16}.}

\begin{proof}
We set $c:=\funref{label_entstandenen}(t)\cdot r.$
{Notice that there is a constant $c_{t},$ depending on $t,$ such that
\change{$c_t = {\cal O}(t\sqrt{\log t})$ and}
if $|E(G)|>c_t\cdot |V(G)|,$ then $G$ contains $K_{t}$ as a minor \cite{Thomason01thee}. We therefore assume that $|E(G)|=\change{{\cal O}(t\sqrt{\log t}\cdot n),}$ otherwise we can immediately report that $K_{t}$ is a minor of $G$  and stop.}
We now describe a recursive algorithm as follows.\medskip

We first argue for the base case, namely when $|V(G)| < 12c^{3}.$
To check whether $K_t$ is a minor of $G,$ we use the minor-containment algorithm of Robertson and Seymour \cite{RobertsonS95b}, which runs in ${\cal O}_t (|V(G)|^3)={\cal O}_t (r^3)$ time, and if this is the case, we report the same and stop.
If not, then we check whether ${\bf tw}(G)\leq c,$ using the algorithm of Arnborg, Corneil, and Proskurowski~\cite{ArnborgCP87comp}, in time ${\cal O}(|V(G)|^{c+2})=2^{2^{{\cal O}(t^2\log t)} r\log r},$ and if this is the case, we report the same and stop.
If not, we deal with the case where $G$ does not contain $K_t$ as a minor and ${\bf tw}(G)>c.$ By
\autoref{label_scommettendo} we know that $G$ contains an $r$-wall.
To find such a wall, we first consider an arbitrary ordering
$(v_1, \dots, v_{|V(G)|})$ of the vertices of $G.$
For each $i\in[|V(G)|],$ we set $G_i$ to be the graph induced by the vertices $v_1, \dots, v_i.$
We iteratively run the algorithm of \autoref{label_naturalistes} on $G_i$ and $c$ for ascending values of $i.$
This algorithm runs in
$2^{{\cal O}(c)}\cdot |V(G)|=
2^{2^{{\cal O}(t^2 \log t)} r}$ time.
Let $j\in[|V(G)|]$ be the smallest integer such that the above algorithm outputs a report that ${\bf tw}(G_j)>c$ and notice that there exists a tree decomposition $({\cal T}_j, \chi_j)$ of $G_j$ (obtained by the one of $G_{j-1}$ by adding the vertex $v_j$ in the appropriate bags) of width at most $5c+5.$
The fact that $G_j$ does not contain $K_t$ as a minor and ${\bf tw}(G_j)>c,$
 implies that $G_j$ contains an $r$-wall $W,$ that is also a wall of $G.$
To detect $W,$
we run the algorithm of \autoref{label_prohibitivas}
on $G_j,$ $W,$ and $({\cal T}_j, \chi_j).$
This algorithm runs in time
$2^{{\cal O}(c\log c) + {\cal O}(c \log r) + {\cal O}(r^2)}\cdot |V(G)|=2^{2^{{\cal O}(t^2\log t)}\cdot r\log r + {\cal O}(r^2)}$.
Therefore, in the case where $|V(G)|\leq 12c^3,$ we obtain one of the three possible outputs in time $2^{2^{{\cal O}(t^2\log t)}\cdot r\log r + {\cal O}(r^2)}+{\cal O}_t (r^3).$
\change{Alternatively,
to get an explicit dependence on $t$, instead of applying the minor-containment algorithm of Robertson and Seymour~\cite{RobertsonS95b} in the beginning of the algorithm,
we can do the following:
first, apply the algorithm of~\autoref{label_naturalistes} on $G$ and $12c^3$.
Since $|V(G)|<12c^3$ and therefore $\tw(G)<12 c^3$, this algorithm outputs a tree decomposition $({\cal T},\chi)$ of $G$ of width $62c^3 +4$.
Then, we apply the algorithm of~\autoref{label_prohibitivas} on $G$, $K_t$, and $({\cal T},\chi)$, to check whether $K_t$ is a minor of $G$, in time $2^{2^{{\cal O}(t^2\log t)} r^3\log r}$.
}

	If  $|V(G)|\geq 12c^{3},$ then we call  the algorithm of \autoref{label_verneinenden} with input $(G,c),$ which outputs a graph $G^{*}$ such that $|V(G^{*})|\leq (1-\frac{1}{16c^{2}}) \cdot |V(G)|$ and
	\begin{itemize}
		\item[{\em A}.] either $G^{*}$ is a subgraph of $G$ such that $\tw(G)=\tw(G^{*}),$ or
		\item[{\em B}.] $G^{*}$ is obtained from $G$ after identifying the vertices of a matching  $M$ of $G.$
	\end{itemize}
	In both cases, we recursively call the algorithm on  $G^{*}$ and we distinguish the following two cases.\medskip

	\noindent{\em Case A}: $G^{*}$ is a subgraph of $G$ such that $\tw(G)=\tw(G^{*}).$ If the recursive call on $G^*$ reports that $K_{t}$ is a minor of $G^*,$ then we report the same for $G$ as well.
	If the recursive call on $G^*$ reports that $\tw(G^{*})\leq c,$ then we return   that $\tw(G)\leq c.$ 	If it outputs an $r$-wall $W$ of $G^{*},$  then we return $W$ as a wall of $G.$ \medskip

	\noindent{\em Case B}:  $G^{*}$ is obtained from $G$ after contacting the edges of a matching of $G.$

	If the recursive call on $G^*$ reports that $\tw(G^{*})\leq c,$  then we do the following.
	We first notice that the fact that $\tw(G^{*})\leq c$ implies that $\tw(G)\leq 2c,$
	since we can obtain a tree decomposition $({\cal T},\chi)$ of $G$ from a tree decomposition $({\cal T}^{*},\chi^{*})$ of $G^{*},$
	by replacing, in every $t\in{\cal T}^{*},$ every occurrence of a vertex of $G^{*}$ that is a result of an edge contraction by its endpoints in $G.$
	Thus, we can call the algorithm of~\autoref{label_prohibitivas} on $G,$ $K_{t},$ and $({\cal T},\chi)$ in order to check whether $G$ contains $K_{t}$ as a minor in $2^{2^{{\cal O}(t^2\log t)}\cdot r\log r}\cdot n$ steps and if this is the case, we report the same and stop (keep in mind that $c=2^{{\cal O}(t^2\log t)}\cdot r$). If not,  then using the same algorithm we can also find in $G,$ if it exists, an $r$-wall $W$ as a minor in $2^{2^{{\cal O}(t^2\log t)}\cdot r\log r + {\cal O}(r^2)}\cdot n$ time  and, if this is the case, we report the same and stop.
	In the remaining case, we can safely report, because of   \autoref{label_scommettendo}, that $\tw(G)\leq \funref{label_entstandenen}(t)\cdot r=c.$

	If the recursive call on $G^*$ outputs an $r$-wall $W^*$ of $G^{*},$  then by
	uncontracting the edges of $M$ in $W^*$ we can also  return an $r$-wall of $G.$
	Finally, if the output is that $K_{t}$ is a minor of $G^*,$ then we return that the same holds for $G.$

	It is easy to see that the running time of the above algorithm is $$T(n,r,t)\ \leq\  T((1-\frac{1}{12c^{2}})\cdot n,r,t)+ 2^{2^{{\cal O}(t^2\log t)}\cdot r\log r + {\cal O}(r^2)}\cdot n,$$
\change{where for $n<12c^3$, $T(n,r,t) = 2^{2^{{\cal O}(t^2\log t)}\cdot r\log r + {\cal O}(r^2)}+{\cal O}_t (r^3) = 2^{{\cal O}_{t}(r^2)}$ or, in the case we ask for an explicit dependence on $t$,
$T(n,r,t) = 2^{2^{{\cal O}(t^2\log t)} r^3\log r}$.
Therefore, we have that $T(n,r,t)=2^{{\cal O}_{t}(r^2)}\cdot n$ or
$T(n,r,t) =
2^{2^{{\cal O}(t^2\log t)} r\log r + {\cal O}(r^2)}\cdot n+2^{2^{{\cal O}(t^2\log t)} r^3\log r},$}
as claimed.
\end{proof}

Given a flatness pair $(W,\frak{R})$ of a graph $G$ and a set $L\subseteq V(G),$ we say that $(W,\frak{R})$ is {\em  $L$-avoiding} if $L\cap V({\sf compass}_{\frak{R}}(W))=\emptyset.$
We now proceed to the proof of \autoref{label_proletarians}.

\begin{proof}[Proof of \autoref{label_proletarians}]
Notice that there is a constant $c_{t},$ depending on $t,$ such that
\change{$c_t = {\cal O}(t\sqrt{\log t})$ and}
if $|E(G)|>c_t\cdot |V(G)|,$ then $G$ contains $K_{t}$ as a minor \cite{Thomason01thee}. We therefore assume that $|E(G)|=\change{{\cal O}(t\sqrt{\log t}\cdot n),}$ otherwise we can immediately report that $K_{t}$ is a minor of $G$  and stop.
We first give  an algorithm with the following specifications.
This algorithm involves recursion assuming an input with an additional set $L$ that should be avoided by the desired flatness pair.
For notational convenience, we define $z:\Bbb{N}^{2}\to\Bbb{N}$
as $z(r,t)=2  \cdot(\lceil\sqrt{\funref{label_hierarchical}(t)+2}\rceil+1)\cdot  \funref{label_entstandenen}(t)\cdot (\funref{label_questionnaires}(t)+1)\cdot (r+2).$
\medskip

\noindent Algorithm {\bf Find\_Low\_TW\_compass}$(G,r,t,L).$

\noindent {\sl \textbf{Input}:} an odd $r \in\Bbb{N}_{\geq 3},$ a  $t\in\Bbb{N}_{\geq 1},$ a graph $G$ where $\tw(G)>  z(r,t),$ and a set $L\subseteq V(G)$ where $|L|\leq \funref{label_hierarchical}(t)+1.$

\noindent {\sl \textbf{Output}:}  either a report  that $K_{t}$ is a minor of $G$ or  a set {$A\subseteq V(G)$}, where $|A|\leq \funref{label_hierarchical}(t),$
an $L$-avoiding
flatness pair $(W,\frak{R})$ of $G\setminus A$ of height
$r$, and a tree decomposition of  the $\frak{R}$-compass of $W$ of  width at most $5\cdot z(r,t)+4.$
\medskip

\noindent{\bf Step 1}. We set $\ell$ as the smallest odd integer that is not smaller than $\sqrt{\funref{label_hierarchical}(t)+2}.$
Also, let $\tilde{\funref{label_questionnaires}}(t)$ be the smallest odd integer that is not smaller than $\funref{label_questionnaires}(t).$ These augmentations are necessary in order to guarantee that the considered subwalls will be of odd height.
We also set $r'=2\cdot (r+2)+1.$
Run the algorithm of \autoref{label_hypostatizing} for $G,$ $\ell\cdot \tilde{\funref{label_questionnaires}}(t)\cdot r',$ and $t.$ This takes
time $2^{{\cal O}_t(r^2)}\cdot n$, \change{or, for an explicit dependence on $t$,
it can be modified to take time $2^{2^{{\cal O}(t^2\log t)} r\log r + {\cal O}(r^2)}\cdot n+2^{2^{{\cal O}(t^2\log t)} r^3\log r}$.}
If the output is a report that $K_{t}$ is a minor of $G,$ then return the same.
Otherwise, because, $\tw(G) >  z(r,t)\geq \ell\cdot  \funref{label_entstandenen}(t)\cdot \tilde{\funref{label_questionnaires}}(t)\cdot  r',$ the algorithm  returns  an $\ell\cdot \tilde{\funref{label_questionnaires}}(t)\cdot  (2(r+2)+1)$-wall $W$ of $G.$

\smallskip

\noindent{\bf Step 2}. Call the algorithm of \autoref{label_inconsiderable} on $G,$ $\ell\cdot r',$ $t,$ and $W.$
% \ig{reorder as in the statement of the theorem}.
This takes \change{${\cal O}(t^{25} \sqrt{\log t}\cdot n)$} time, since $|E(G)|={\cal O}(t\sqrt{\log t}\cdot n).$ If the   output  is a report that $K_{t}$ is a minor of $G,$ then return the same. Otherwise, we have  a set $A\subseteq V(G)$, where $|A|\leq \funref{label_hierarchical}(t),$ and a
flatness pair $(\tilde{W}',\tilde{\frak{R}}')$ of $G\setminus A$ of height $\ell\cdot r'.$
\smallskip

\noindent{\bf Step 3}.
	Let $W''$ be a subwall of $\tilde{W}'$
	of height $r'$ such that none of the vertices in $L$ belongs to
	${\sf influence}_{\tilde{\frak{R}}'}(W'').$
	The subwall $W''$
	exists because $\ell^2\geq  \funref{label_hierarchical}(t)+2\geq |L|+1$ and $\tilde{W}'$ has height  $\ell\cdot r'.$
	We also consider four pairwise disjoint $(r+2)$-subwalls of ${W}'',$ namely  $W_1',$ $W_{2}',$ $W_{3}',$ and $W_{4}',$ and observe that each $W_i'$ is also a subwall of $\tilde{W}'.$
	For every $i\in[4],$ we call the algorithm of \autoref{label_proporcionada} on $G\setminus A,$ $(\tilde{W}',\tilde{\frak{R}}'),$ and $W_i'$
	which outputs, in ${\cal O}_{t}(n)$ time, a $W_i '$-tilt $(\tilde{W}_{i}',\tilde{\frak{R}}_{i}')$ of $(\tilde{W}',\tilde{\frak{R}}').$
	Let $K'_i$ be the compass of $\tilde{W}_{i}'$ in $\tilde{\frak{R}}_{i}'.$
	We finally fix $i$ so that $\tilde{W}_{i}'$ is a wall among $W_1',$ $W_{2}',$ $W_{3}',$ and $W_{4}'$ where $|V(K_{i}')|$ is minimized.
	Observe that $|V(K_{i}')|\leq |V(G)|/4$ and that   $(\tilde{W}'_{i},\tilde{\frak{R}}'_{i})$ is $L$-avoiding.
	Indeed, since $(\tilde{W}_{i}',\tilde{\frak{R}}_{i}')$ is a $W_i '$-tilt of $(\tilde{W}',\tilde{\frak{R}}'),$
	$K_{i}'={\sf compass}_{\tilde{\frak{R}}_{i}'}(\tilde{W}_{i}')$
	is a subgraph of $\cupall{\sf influence}_{\tilde{\frak{R}}'}(W'_{i})$ that, in turn, is a subgraph of $\cupall{\sf influence}_{\tilde{\frak{R}}'}({W}'')$
	and by definition of $W'',$ ${\sf influence}_{\tilde{\frak{R}}'}({W}'')\cap L=\emptyset.$

	We update
	$W\leftarrow \tilde{W}'_{i},$ $\frak{R}\leftarrow \tilde{\frak{R}}'_{i}$
	and we set $K={\sf compass}_{\frak{R}}(W).$ Recall that
	$(W,\frak{R})$ is an $L$-avoiding flatness pair of $G\setminus A$ of height $r+2.$

	%such that none of the vertices in $L$ belongs in some of~ $\cupall{\sf cfl}_{\frak{R}}(D(W^{}_{i})), i\in[4].$  Notice that, as $|L|\leq \funref{label_hierarchical}(t)+1$ and \sed{First find a $2(r+2)$ subwall and then find the 4 walls inside}
	%$\ell=2\cdot\lceil\sqrt{\funref{label_hierarchical}(t)+1}\rceil,$ we can assume that $r^*=r+2.$

	%Then, for $i\in[4],$ let  $\theta(W_{i})=(A,W'_{i},\frak{R}'_{i})$ be an $(\funref{label_hierarchical}(t),r+2)$-apex-wall of $G$ whose compass is $K'_i$ and whose tilt-assignment function is $\theta'_i.$
	%Let $W_{i}'$ be the wall where $|V(K_{i}')|$ is minimized.
	%Observe that $|V(K_{i}')\leq |V(G)|/4$ and that   $(A,W'_{i},\frak{R}'_{i})$ is $L$-avoiding because ${\sf compass}_{\frak{R}_{i}'}(W_{i}')\subseteq \cupall{\sf cfl}_{\frak{R}}(D(W^{}_{i})), i\in[4].$

	\smallskip

	\noindent{\bf Step 4}.
	We now consider the subwall $W'$ of $W$
	obtained from $W\setminus D(W)$ after repeatedly
	removing vertices of degree one until no such vertices  exist anymore. Notice that $W'$ is an $r$-wall of $G\setminus A.$ We call the algorithm of \autoref{label_proporcionada} on $G\setminus A,$ $(W,\frak{R}),$ and $W'$
	which outputs, in \change{${\cal O}(t\sqrt{\log t}\cdot n)$} time, a $W'$-tilt $(\tilde{W}',\tilde{\frak{R}}')$ of $(W,\frak{R}).$
	Let $K'$ be the $\tilde{\frak{R}}'$-compass of $\tilde{W}'.$
	%   and let $\theta'$ \ig{$\theta'$ is not used anywhere in this proof}
	% be the tilt-assignment function of $(\tilde{W}',\tilde{\frak{R}}').$
	Clearly, $(\tilde{W}',\tilde{\frak{R}}')$ is $L$-avoiding as well.

	\smallskip

	\noindent{\bf Step 5}. Let $G_D$  be the graph obtained from $G[V(K)\cup A]$ if we contract all the vertices
	of $D(W)$ to a single vertex $v^*.$
	Since $(\tilde{W}',\tilde{\frak{R}}')$ is a $W'$-tilt  of $(W,\frak{R}),$ $K'={\sf compass}_{\tilde{\frak{R}}'}(\tilde{W}')$ is a subgraph of ~$\cupall{\sf influence}_{{\frak{R}}}(W'),$ and therefore the perimeter of $W$ and the graph $K'$ do not have any vertex in common. This implies that $K'$ is a subgraph of  $G_D.$

	%
	%
	%\noindent{\bf Step 5}. Call the algorithm of \autoref{label_naturalistes} on $G_D$ and $r.$
	%If $\tw(G_D)>r,$ we set  $G\leftarrow G_D,$ $L= A\cup \{v'\}$
	%and goto Step 2.
	%

	\smallskip

	\noindent{\bf Step 6}. Call the algorithm of \autoref{label_naturalistes} with input  $G_{D}$ and
	$z(r,t).$ This runs in \change{$2^{2^{{\cal O}(t^2\log t)}\cdot r}\cdot n$} time.
	If the output is a tree decomposition of $G_D$
	of width at most $5\cdot z(r,t)+4,$  then, as $K'$ is a subgraph of  $G_D,$
	we have that  $(\tilde{W}',\tilde{\frak{R}}')$ is an $L$-avoiding
	flatness pair of $G\setminus A$ of height $r$  where the $\tilde{\frak{R}}'$-compass of $\tilde{W}'$ has treewidth at most $5\cdot z(r,t)+4.$  In this case, the algorithm outputs the pair $(\tilde{W}',\tilde{\frak{R}}')$ and the corresponding tree decomposition of the  $\tilde{\frak{R}}'$-compass $K'$ of $\tilde{W}'$ obtained from the one of $G_{D}$ by removing the vertices in $V(G_{D})\setminus V(K').$

	\smallskip

	\noindent{\bf Step 7}.  Suppose now that $\tw(G_D)> z(r,t).$
	Notice that, by construction, if $G_D\setminus A$ has an $\{v^*\}$-avoiding flatness pair $(W^*,\frak{R}^*)$ of height $r,$ then $(W^*,\frak{R}^*)$ will also be an $L$-avoiding  flatness pair  of $G\setminus A.$
	Moreover, since $G_{D}$ is a minor of $G,$ if $G_D$ contains $K_{t}$ as a minor then also $G$ does.
	Notice also that  $|A\cup\{v'\}|\leq \funref{label_hierarchical}(t)+1.$
	Therefore, we can safely return {\bf Find\_Low\_TW\_compass}$(G_D,r,t,A\cup\{v'\}).$
	This completes the description of the algorithm and its correctness.\medskip

Notice  that the running time of the above algorithm is
$$T(n,r,t)\ \leq  T(n/4+\funref{label_hierarchical}(t),r,t)+2^{{\cal O}_t(r^2)}\cdot n,$$
which implies that $T(n,r,t)=2^{{\cal O}_t(r^2)}\cdot n,$
\change{and can be modified in order to obtain $T(n,r,t) =2^{2^{{\cal O}(t^2\log t)} r\log r + {\cal O}(r^2)}\cdot n+2^{2^{{\cal O}(t^2\log t)} r^3\log r}.$}
\medskip

	We  define the function $\funref{label_confrontation}:\Bbb{N}\to\Bbb{N}$ so that $\funref{label_confrontation}(t)=\min\{c\in \Bbb{N}\mid \forall r\geq 3,\  5\cdot z(r,t)+4\leq c\cdot  r\}.$
	The algorithm claimed by the theorem calls first the algorithm of \autoref{label_naturalistes} with input $G$ and $z(r,t).$ This runs in  \change{$2^{2^{{\cal O}(t^2\log t)}\cdot r}\cdot n$} time.
	If the output is a tree decomposition of $G$ of width at most $5\cdot z(r,t)+4\leq\funref{label_confrontation}(t) \cdot r,$ then we report this and we are done.
	%
	%If the output is a report that $\tw(G)>z(r,t),$ then we return the output of Algorithm {\bf Find\_Low\_TW\_compass}$(G,r,t,L)$
	%for $L=\emptyset.$ This may provide either a report  that $K_{t}$ is a minor of $G,$ or
	% a set {$A\subseteq V(G)$} where $|A|\leq \funref{label_hierarchical}(t),$
	%a flatness pair $(W,\frak{R})$ of $G\setminus A$ of height
	%$\funref{label_hierarchical}(t),$  and a tree decomposition of  the $\frak{R}$-compass of $W$ of  width at most $5\cdot z(r,t)+4\leq\funref{label_confrontation}(t) \cdot r,$ as required.
	If the output is a report that $\tw(G)>z(r,t),$ then we run Algorithm {\bf Find\_Low\_TW\_compass}$(G,r,t,L)$
	for $L=\emptyset.$ This may provide either a report  that $K_{t}$ is a minor of $G,$ or
	a set {$A\subseteq V(G)$}, where $|A|\leq \funref{label_hierarchical}(t),$
	a flatness pair $(W,\frak{R})$ of $G\setminus A$ of height
	$\funref{label_hierarchical}(t)$ that can be made regular by \autoref{label_considerabil},  and a tree decomposition of  the $\frak{R}$-compass of $W$ of  width at most $5\cdot z(r,t)+4\leq\funref{label_confrontation}(t) \cdot r,$ and these are the possible outputs of the claimed algorithm.
	%In the first case, we return the same, while in the second one, we run the algorithm of \autoref{label_considerabil} on $G\setminus A$ and $(W,\frak{R})$ which, in ${\cal O}_{t}(n)$ time, outputs a regular flatness pair $(W^{\star},\frak{R}^{\star})$ of $G \setminus A$ of height $\funref{label_hierarchical}(t)$ such that ${\sf compass}_{\frak{R}^{\star}}(W^{\star})\subseteq {\sf compass}_{\frak{R}}(W).$
	%Moreover, we return the  corresponding tree decomposition of the $\frak{R}^{\star}$-compass of $W^{\star}$ obtained from the one of the $\frak{R}$-compass of $W$ by removing the vertices in $V({\sf compass}_{\frak{R}}(W))\setminus V({\sf compass}_{\frak{R}^{\star}}(W^{\star})).$
\end{proof}

\subsection{{Homogeneous walls}}\label{label_establecidas}
%We first present some definitions on boundaried graphs  that will be used to define the notion of homogeneous walls.

\paragraph{Palettes and homogeneity.}

Let $w\in \Bbb{N},$
let $G$ be a graph, and let $(W,\frak{R})$ be a  flatness pair of  $G.$ A {\em flap-coloring of $(W,\frak{R})$
with $w$ colors} is any function $ζ: {\sf flaps}_{\frak{R}}(W)\to [w].$
For every $\frak{R}$-normal cycle $C$ of ${\sf compass}_{\frak{R}}(W),$ we define $\zeta\mbox{\sf -palette}(C)=\{\zeta(F)\mid F\in {\sf  influence}_{\frak{R}}(C)\}.$
	We say that  the flatness pair $(W,\frak{R})$  of $G$  is {\em $\zeta$-homogeneous}    if  every  {\sl internal} brick of ${W}$   (seen as a cycle of ${\sf compass}_{\frak{R}}(W)$) has the {\sl same} $\zeta$\mbox{\sf -palette}.

Finding a homogeneous flatness pair inside a flatness pair has a price which is determined by the following lemma.

\begin{lemma}
	\label{label_legisladores}
	There is a function $\newfun{label_eigenschaftswort}: \Bbb{N}^{2}\to\Bbb{N},$ whose images are odd integers,
	such that for every $w\in\Bbb{N}_{\geq 1}$ and every odd integer $r\geq 3,$ if  $G$ is a graph, $(W,\frak{R})$ is a flatness pair of $G$ of height $\funref{label_eigenschaftswort}(r,w),$
and  $\zeta$ is a flap-coloring of $(W,\frak{R})$
with $w$ colors,
then $W$ contains some subwall $W'$ of height $r$ such that every $W'$-tilt of $(W,\frak{R})$ is $\zeta$-homogeneous.
Moreover, $\funref{label_eigenschaftswort}(r,w)={\cal O}(r^w).$
%and ${\sf compass}_{\tilde{\frak{R}}'}(\tilde{W}')$ is a subgraph of ${\sf compass}_{\frak{R}}(W).$
\end{lemma}

\begin{proof}
Let $w\in \Bbb{N}$ and an odd integer $r\geq 3.$
We define the function $\funref{label_eigenschaftswort}: \Bbb{N}^{2}\to\Bbb{N}$ so that,
for every $x\in\Bbb{N},$  $\funref{label_eigenschaftswort}(x,1)=x$ while,
for $y\geq 2,$ we set  $\funref{label_eigenschaftswort}(x,y)=x\cdot (\funref{label_eigenschaftswort}(x,y-1)-1)+1.$ Notice that if $x$
is  odd, then $\funref{label_eigenschaftswort}(x,y)$
is also odd for every $y\in\Bbb{N}_{\geq 1}.$

Let $G$ be a graph, $(W,\frak{R})$ be a  flatness pair  of  $G$ of height $\funref{label_eigenschaftswort}(r,w),$ and $\zeta$ be a flap-coloring of $(W,\frR)$ with $w$ colors.
We prove the lemma by induction on $w.$
Clearly, if  $w=1,$ then  the lemma holds trivially as, in this case,
for every brick $B$ of $W,$  $\zeta\mbox{\sf -palette}(B)=\{1\},$
and therefore as $W$ is a subwall of itself,
every $W$-tilt of $(W,\frak{R})$  is a flatness pair of $G$ of height $\funref{label_eigenschaftswort}(r,1)=r$ that is  $\zeta$-homogeneous.

Suppose now that $w\geq 2$ and that the lemma holds for smaller values of $w.$
We set $q=\funref{label_eigenschaftswort}(r,w-1).$
We define the subwall ${W}'$ of $W$ by taking
the union of the $i$-th horizontal and the $i$-th vertical paths of $W$ for all $i\in\{j\cdot (q-1)+1\mid j\in[r]\}.$
If for every brick $B$ of  ${W}'$ it holds that $\zeta\mbox{\sf -palette}(B)=[w],$ then consider a $W'$-tilt $(\tilde{W}',\tilde{\frak{R}}')$ of $(W,\frak{R}).$
The third property in the definition of a tilt of a flatness pair
%	the set of $\tilde{W}'$-internal  cells of  $\tilde{\frak{R}}'$ is the same as the set of $W'$-internal cells of ${\frak{R}}$ and their images via $\sigma'$ and ${\sigma}$ are also the same	
implies that for every  internal brick $\tilde{B}$ of $\tilde{W}'$ there is an internal brick $B$ of $W'$ such that ${\sf influence}_{\frak R}(B)={\sf influence}_{\tilde{\frak R}'}(\tilde{B}).$
Therefore, for every internal brick $\tilde{B}$ of $\tilde{W}',$ $\zeta\mbox{\sf -palette}(\tilde{B})=[w].$
Therefore,	 $(\tilde{W}',\tilde{\frak{R}}')$ is a flatness pair of $G$ of height $r$ that is $\zeta$-homogeneous.
Otherwise, let $\breve{B}$ be some
brick of $W'$ such that $|\zeta\mbox{\sf -palette}(\breve{B})|< w.$
Notice that  $\breve{B}$  is the perimeter of a subwall $\breve{W}$ of $W$
of height $q.$
From the induction hypothesis applied to  $\breve{W},$ we have that
$\breve{W}$ has a subwall  $W'$ (that is a subwall of $W$ as well) such that every $W'$-tilt of $(W,\frak{R})$  is a flatness pair of $G$ of height $r$ that is  $\zeta$-homogeneous.
The lemma follows by observing that $\funref{label_eigenschaftswort}(r,w)={\cal O}(r^{w}).$
\end{proof}

{
We now prove the  main result of this subsection.
\begin{lemma}\label{label_desvanecidos}
There is an algorithm that receives as  input  $w\in\Bbb{N}_{\geq 1},$ an odd integer $r\geq 3,$  a graph $G,$ a flatness pair $(W,\frR)$ of $G$ of height $\funref{label_eigenschaftswort}(r,w),$ and a flap-coloring $ζ$ of
$(W,\frR)$ with $w$ colors,
and outputs
a  $\zeta$-homogeneous  flatness pair $(\breve{W},\breve{\frR})$ of $G$  of height $r$ that is a $W'$-tilt of $(W,\frR)$ for some subwall $W'$ of $W.$
This algorithm runs in time {$2^{{\Ocal }(w r \log r)}\cdot(n+m)$}.
\end{lemma}

\begin{proof}
Let ${\cal W}$ be the collection of all $r$-subwalls of $W.$ Clearly $|{\cal W}|={\binom{\funref{label_eigenschaftswort}(r,w)}{r}}^{2}=2^{{\Ocal }(w r \log r)}.$
For each $W'\in {\cal W},$ we call the algorithm of \autoref{label_proporcionada} on $G,$ $(W,\frak{R}),$ and $W',$ which outputs, a $W'$-tilt $(\tilde{W}',\tilde{\frak{R}}')$ of $(W,\frak{R}).$
This algorithm runs in ${\cal O}(n+m)$ time.
Then, for every $W'\in {\cal W},$ we check whether
$(\tilde{W}',\tilde{\frak{R}}')$ is $\zeta$-homogeneous by computing the $\zeta$-${\sf palette}(\tilde{B})$ for every internal brick $\tilde{B}$ of $\tilde{W}'.$
This is done in linear time.
\autoref{label_legisladores} guarantees that since the height of $(W,\frR)$ is $\funref{label_eigenschaftswort}(r,w),$  $W$ contains a  subwall $W'$
of height $r$ such that every $W'$-tilt of $(W,\frak{R})$
is $\zeta$-homogeneous.
Therefore, the above procedure will detect a flatness pair $(\tilde{W}',\tilde{\frak{R}}')$ of $G$ that is $\zeta$-homogeneous and has height $r,$ which we return.
\end{proof}
}

\subsection{Levelings and well-aligned flatness pairs}
\label{label_definitionen}

Let $G$ be a graph and let $(W,\frak{R})$ be a flatness pair of $G.$
Let also  $\frak{R}=(X,Y,P,C,\Gamma,\sigma,\pi),$  where $(\Gamma,\sigma,\pi)$ is an $\Omega$-rendition of $G[Y]$ and $\Gamma=(U,N)$ is a $\Delta$-painting.
%We call $(X,Y)$ the {\em separation certifying the flat wall $W$} and $X\cap Y$ is called the  {\em frontier} of $W.$
The {\em ground set} of $W$ in ${\frak{R}}$ is ${\sf ground}_{\frak{R}}(W):=\pi(N(\Gamma))$ and we refer to the vertices of this set as the {\em ground vertices} of the $\frak{R}$-compass of $W$ in $G.$
Notice  that ${\sf ground}_{\frak{R}}(W)$ may  contain vertices
of ${\sf compass}_{\frak{R}}(W)$ that are not necessarily vertices of $W.$
For instance, in  \autoref{label_consumadamente}, all the ground vertices of the $\tilde{\frak{R}}'$-compass of $\tilde{W}'$ are vertices of $\tilde{W}',$ while in  \autoref{label_exhalaciones}, there are ground vertices
of the $\frak{R}$-compass of ${W}$
that are not vertices of $W.$

\begin{figure}[h]
	\begin{center}
		\reflectbox{\includegraphics[width=12cm]{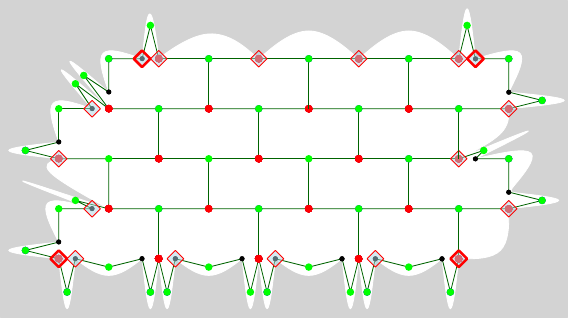}}
	\end{center}
	\caption{The  $\tilde{\frak{R}}'$-leveling of the flat 5-wall {$\tilde{W}'$} of \autoref{label_consumadamente}.}
	\label{label_circustancias}
\end{figure}

We define  the $\frak{R}$-{\em leveling}  of $W$ in $G,$
denoted by ${W}_{\frak{R}},$ as the bipartite graph
where  one part is the ground set of $W$ in $\frak{R},$ the  other part is a set ${\sf vflaps}_{\frak{R}}(W)=\{v_{F}\mid F\in {\sf flaps}_{\frak{R}}(W)\}$ containing one new vertex $v_{F}$ for each flap  $F$ of $W$ in $\frak{R},$
and, given  a pair $(x,F)\in {\sf ground}_{\frak{R}}(W)\times {\sf flaps}_{\frak{R}}(W),$   the set $\{x,v_F\}$ is an edge of ${W}_{\frak{R}}$ if and only if
$x\in \partial F.$ We call the vertices of ${\sf ground}_{\frak{R}}(W)$ (resp. ${\sf vflaps}_{\frak{R}}(W)$) {\em ground-vertices} (resp. {\em flap-vertices}) of ${W}_{\frak{R}}.$
%Again, keep in mind that ${W}_{\frak{R}}$ may contain (many) vertices that are not in $W.$
Notice that the incidence graph of the plane hypergraph $(N(\Gamma),\{\tilde{c}\mid c\in C(\Gamma)\})$ is isomorphic to ${W}_{\frak{R}}$
via an isomorphism that extends  $\pi$ and, moreover, bijectively corresponds cells to flap-vertices.
This permits us to treat ${W}_{\frak{R}}$ as a $\Delta$-embedded graph where  $\bd(\Delta)\cap {W}_{\frak{R}}$ is the set $X\cap Y.$
As an example,  see \autoref{label_circustancias} for the  $\tilde{\frak{R}}'$-leveling of the flat 5-wall {$\tilde{W}'$} of \autoref{label_consumadamente}.
%The graph $W^\bullet$ can be seen as a ``slight variant'' of $W$  that will be important at the end of this subsection for representing
%$W$ by a wall in $W_{\frak{R}}.$

{We denote by $W^{\bullet}$ the  graph obtained from $W$ if we subdivide {\em once} every
edge of $W$ that is {short} in ${\sf compass}_{\frak{R}}(W).$}
The graph $W^\bullet$  is  a ``slightly richer variant'' of $W$  that is necessary for our definitions and  proofs, namely to be able to associate  every flap-vertex of  an appropriate subgraph of $W_{\frak{R}}$ (that we will denote by $R_{W}$) with  a non-empty path of $W^\bullet,$ as we proceed to formalize.
We say that $(W,\frak{R})$ is {\em well-aligned} if the following holds:
\begin{quote}
	$W_{\frak{R}}$ contains as a subgraph an $r$-wall $R_{W}$
	where ${D(R_{W})}=D({W}_{\frak{R}})$ and $W^{\bullet}$  is isomorphic to some subdivision of  $R_{W}$
	via an isomorphism that maps each ground vertex to itself.
\end{quote}
Suppose now that the flatness pair $(W,\frak{R})$ is well-aligned.
We call the wall  $R_{W}$ in the above condition  a {\em representation} of $W$ in $W_{\frak{R}}.$

As an example, notice that the flatness pair $(\tilde{W}',\tilde{\frak{R}}')$ of \autoref{label_consumadamente}  is well-aligned  while the  flatness pair $(W,\frak{R})$ in \autoref{label_exhalaciones} is not
	{since, for example, in the uppermost rightmost grey cell, the upper right ground vertex can not be mapped to itself in order to yield a subgraph $R_W$ of ${\cal W}_{\frak{R}}$ as in the above property.}
%Later we will see how this property can be imposed (see \autoref{label_considerabil}).

\remove{Notice that both $W_{\frak{R}}$ and its subgraph $R_{W}$ can be seen as $\Delta$-embedded graphs where $\bd(\Delta)\cap {W}_{\frak{R}}=\bd(\Delta)\cap R_{W}\subseteq V(D(W_{\frak{R}}))=V(D(R_{W})).$
This permits us  to set up a bijective map from  each cycle {$C$} of $W$  to a cycle ${\sf rep}_{\frak{R}}(C)$ of $R_{W}.$
%{ and from each subwall ${W}'$ of $W$ to a subwall ${\sf rep}_{\frak{R}}({W}')$ of $R_{W}.$}
We also  define the function ${\sf cfl}_{\frak{R}}: {\cal C}(W)\to 2^{{\sf flaps}_{\frak{R}}(W)}$  so that,  for each cycle $C$ of $W,$ ${\sf cfl}_{\frak{R}}(C)$ contains each flap $F$ of  $W$ in $\frR$ so that  $v_{F}$
belongs to  the closed disk of $\Delta$ bounded by ${\sf rep}_{\frak{R}}(C).$  The following observation follows directly from the definitions.

\begin{observation}
	If $(W,\frak{R})$ is a well-aligned flatness pair and $W'$ is a
	wall  of $W,$ then ${\sf cfl}_{\frak{R}}(D(W'))={\sf influence}_{\frak{R}}(W').$
\end{observation}
}

\begin{lemma}
	\label{label_indisdinctively}
	If a flatness pair $({W},\frak{{R}})$  is regular, then it is also well-aligned.
	Moreover, there is an {$\Ocal(n)$} time algorithm that, given $G$ and such a $({W},\frak{{R}}),$ outputs a representation $R_{W}$
	of $W$ in $W_{\frak{R}}.$
\end{lemma}

\begin{proof}%[Proof of \autoref{label_indisdinctively}.]
	Let  $({W},\frak{{R}})$  be a  flatness pair where all cells of $\frak{R}$ are tidy and with no
	$W$-external or $W$-marginal cells. We claim that none of the cells of  $\frak{R}$
	is $W$-outer-perimetric. Indeed, a  $W$-outer-perimetric $c$ should correspond to one of the tree last cases of \autoref{label_rigoureusement} (this figure appears later in \autoref{sec_classification_cells} in order to illustrate further definitions): in the fifth case $c$ is untidy and  in the sixth and seventh case $c$ is $W$-marginal.
	Therefore all cells are either $W$-internal or $W$-inner-perimetric and are also all tidy.

	We also denote
	$\frak{R}=(X,Y,P,C,\Gamma,\sigma,\pi).$  Recall that $W^{\bullet}$ (whose edges are depicted in orange  in \autoref{label_entgegengesetzt}) is the  graph obtained from $W$ if we subdivide once every short
	edge in $W.$
	Let $\xi$ be the function mapping every vertex created by a subdivision of a short edge of $W^\bullet$
	(depicted by a cross in \autoref{label_entgegengesetzt}) to the
	corresponding (trivial) flap-vertex of $W_{\frak{R}}$ (that is depicted as one of the blue vertices of degree two).

	Consider $R_{W}=(B\cup F_{1}\cup F_{2},E'),$ where
	\begin{eqnarray*}
		B & =& W\cap {\sf ground}_{\frak{R}}(W),\\
		F_{1}& = & \{\xi({\sf x})\mid \mbox{${\sf x}$ is a subdivision vertex of $W^\bullet$}\}, \text{ and}  \\
		F_{2} & =& \{v_{F}\in{\sf vflaps}_{\frak{R}}(W)\mid E(W\cap F)\neq\emptyset \mbox{~and $F$ is a non-trivial flap}\}.
	\end{eqnarray*}
	In \autoref{label_entgegengesetzt}, the vertices in $B$ are depicted in red in \autoref{label_entgegengesetzt} while the vertices in ${ F}_{1}\cup { F}_{1}$ are depicted in blue.
	We define $E'$ as follows.
	For every $v_F\in { F}_{1}$ we include in $E'$ both edges of $W_{\frak{R}}$ that incident to $v_F.$ For every $v_{F}\in { F}_{2}$ such that $F\setminus 	 \partial F$ contains a 3-branch vertex of $W$ we include in $E'$ the three edges of $W_{\frak{R}}$ that incident to $v_F.$ Finally, for every $v_F\in { F}_{2}$ such that $F\setminus \partial F$ does not contain any 3-branch vertex of $W$ we first consider the non-trivial path $P_{F}$ in
	$W\cap F$ and we add in $E'$ the edges of $W_{\frak{R}}$ between the flap-vertex  $v_F$ and the endpoints of $P_{F}.$
	Notice  that since $\sigma^{-1}(F)$ is tidy, $P_{F}$ does not contain internal vertices in $\partial F.$
	Observe that $R_{W}$ is indeed a wall of ${W}_{\frak{R}},$ where $D({W}_{\frak{R}})=D(R_{W}),$  that can be computed in $\Ocal(n)$ time.
	We now define a mapping $\rho: V(R_{W})\to V(W^{\bullet})$ and a function $\tau$ mapping the edges in  $E(R_{W})$ (depicted as fat purple edges in \autoref{label_entgegengesetzt}) to subpaths of $W^{\bullet}$ as follows:

	\begin{figure}[h]
		\begin{center}
			\includegraphics[width=14cm]{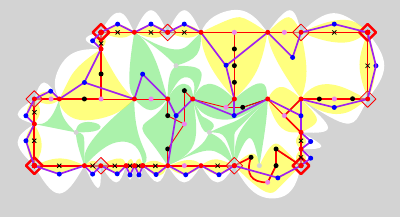}
		\end{center}		\vspace{-2mm}
		\caption{A well-aligned flatness pair $(W,\frak{R})$  where $W$ is a $3$-wall, the wall $W^{\bullet}$ (whose edges are depicted in red   and the new subdivision vertices are depicted by small crosses), the leveling ${W}_{\frak{R}}$ of $W$ (whose edges are depicted in purple), and the subgraph $R_{W}$ of ${W}_{\frak{R}}$
			(depicted by fat purple edges).}
		\label{label_entgegengesetzt}
	\end{figure}

	\begin{itemize}
		\item  If $x\in B,$ then $\rho(x)=x.$

		\item  If $v_F\in { F}_{1}$ and $\partial F=\{x,y\},$ then we set $\rho(v_F)=\xi^{-1}(v_F),$ $\tau(\{x,v_F\})=\{x,\xi^{-1}(v_F)\},$ and $\tau(\{y,v_F\})=\{y,\xi^{-1}(v_F)\}.$

		\item If 	$v_F\in { F}_{2}$ and $v_F$ is  a {branch} vertex of $R_{W},$ then assume first that $\partial F=\{x,y,z\}.$
		      Because the cell $\sigma^{-1}(F)$
		      is tidy the graph $F\setminus \partial F$ contains a unique 3-branch vertex $w$ of $W$
		      (or equivalently of $W^\bullet$) and  $F\cap W^\bullet$ consists of three internally disjoint paths $P_{w,x},$ $P_{w,y},$ and $P_{w,z}$ in $F$ from $w$ to $x,$ $y,$ and $z,$ respectively.  We set $\rho(v_F)=w,$ $\tau(\{x,v_F\})=P_{w,x},$ $\tau(\{y,v_F\})=P_{w,y},$ and $\tau(\{z,v_F\})=P_{w,z}.$

		\item If 	$v_F\in {F}_{2}$ and $v_F$ is not a 3-branch vertex of $R_{W},$ then there exist two vertices $x,y$ of $R_{W}$ such that  $N_{R_{W}}(v_F)=\{x,y\}.$ Pick an internal vertex $w$ of  the $(x,y)$-path $P_{F}$ and set $\rho(v_F)=w$ (recall that, as $\sigma^{-1}(F)$ is tidy, none of the internal vertices of the path $P_{F}$ is a ground vertex).
		      If $P_{w,x}$ is the $(w,x)$-subpath of $P_{F},$ and $P_{w,y}$ is the $(w,y)$-subpath of $P_{F},$ then set $\tau(\{x,v_F\})=P_{w,x}$ and $\tau(\{y,v_F\})=P_{w,y}.$

	\end{itemize}
	It is now easy to verify that the  mappings $\rho$ and $\tau$ defined above certify that $W^{\bullet}$
	is isomorphic to a subdivision of $R_{W}$ by an isomorphism extending $\rho$  (see \autoref{label_entgegengesetzt} for an example).  As all members of $B=W\cap {\sf ground}_{\frak{R}}(W)$  are, by definition, fixed points of $\rho,$ then $({W},\frak{{R}})$ is well-aligned.
\end{proof}

\section{Proofs of  \autoref{label_proporcionada} and  \autoref{label_considerabil}}
\label{dsanfldfalksdsa}
This section is devoted to the proofs of \autoref{label_proporcionada} and  \autoref{label_considerabil}. We first present some definitions in \autoref{sec_stretchings} and
\autoref{sec_classification_cells}, necessary for the proof of the main technical lemma of this paper, namely \autoref{label_weltverbesserer}, presented in \autoref{label_desampararos}.

\subsection{Stretchings}
\label{sec_stretchings}

Let $F$ be a graph and $x$ and $y$ be
two distinct vertices belonging to the same connected component of $F.$
We say that a sequence
$\langle F_{1},\ldots,F_{r}\rangle$ of subgraphs of $F$ is  a {{\em stretching of $F$ along the pair $(x,y)$}} if there is a shortest $(x,y)$-path  $P_{F}$  in $F$ such that
%\ig{this is ambiguous, as the stretching does not only depend on the pair $(y,z),$ but also on the choice of the shortest path $P_F$}
the sequence  $\langle F_{1},\ldots,F_{r}\rangle$ consists of the (unique) minimum-sized collection of subpaths of $P_F$
with the following properties:

\begin{itemize}
	\item each path in $\langle F_{1},\ldots,F_{r}\rangle$ is  a path where all internal vertices have degree two in $F,$
	\item no two paths in $\langle F_{1},\ldots,F_{r}\rangle$ have a common edge,
	\item $F_{1}\cup\cdots\cup F_{r}=P_{F},$
	\item for every $(i,j)\in{[r]\choose 2},$ $F_{i}\cap F_{j}\neq\emptyset$ if and only if  $|i-j|=1,$ and
	\item $x\in V(F_{1})$ and $y\in V(F_{r}).$
\end{itemize}

For an example of a streching of a graph $F$ along a pair $(x,y),$ see \autoref{label_eigentliches}.

\begin{figure}[h]
	\centering
	\includegraphics[width=8cm]{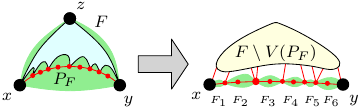}
	\caption{The stretching of a graph $F$ along the pair $(x,y).$}
	\label{label_eigentliches}
	\vspace{-2mm}
\end{figure}

\subsection{Classifying perimetric cells}
\label{sec_classification_cells}

Let $G$ be a graph and let $({W},{\frak{R}})$ be a flatness pair of $G,$ where ${\frak{R}}=(X,Y,P,C,\Gamma,\sigma,\pi).$
Let $W'\in{\cal S}_{{\frak{R}}}({W}).$
We now further refine the classification of the cells of $\frak{R}$ that we gave in~\autoref{label_prognostication}
with respect to $W'.$
See \autoref{label_rigoureusement} for an illustration of the ways a $W'$-perimetric cell $c$ of $\Gamma$ may intersect $\Delta_{W'}.$
The simplest case if when $|\tilde{c}|=2,$ depicted in the leftmost configuration of the figure. The remaining configurations correspond to the case where  $\partial\sigma(c)=\{x,y,z\}$ where $A_{c}$ is a $(\pi^{-1}(x),\pi^{-1}(y))$-arc
(see \autoref{label_prognostication} for the definition of the paths $P^{\rm in}_{c}$ and $P^{\rm out}_{c},$ the arc $A_c,$ and the vertex $z$).
The second/fifth, third/sixth, and forth/seventh configurations correspond to the case where $z$ is an
internal vertex of $P^{\rm in}_{c},$ $P^{\rm out}_{c},$ or none of them, respectively.
This permits  a further classification of the $W'$-perimetric cells of $\Gamma$ as follows.
A cell $c$ of $\Gamma$ is {\em $W'$-inner-perimetric} (resp. {\em $W'$-outer-perimetric}) if  $c\cap \Delta_{W'}$ is situated in $c$ as indicated in the left (resp. right)  part of \autoref{label_rigoureusement}.
%A cell of ${\frak{R}}$ is called {\em inner/outer-perimetric} if it is   $W'$-inner/outer-perimetric.

\begin{figure}[H]
	\begin{center}
		\includegraphics[width=13cm]{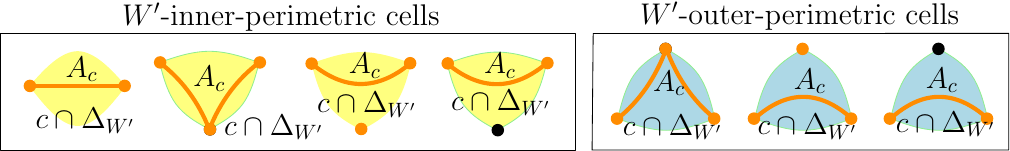}
	\end{center}		\vspace{-2mm}
	\caption{{Seven}  ways $\Delta_{W'}$ may traverse a cell. The arc $A_c$ is depicted in orange.}
	\label{label_rigoureusement}
\end{figure}

We  denote the  set of cells  of $\Gamma$
that are  $W'$-inner-perimetric, $W'$-outer-perimetric, $W'$-internal, and $W'$-strictly external by   $C_{W'}^{\sf ip}(\Gamma), C_{W'}^{\sf op}(\Gamma), C_{W'}^{\sf in}(\Gamma),$ and $C_{W'}^{\sf ex}(\Gamma),$ respectively.
See \autoref{label_simoneggiando} for an example of this further classification (relatively to \autoref{label_exhalaciones}). Notice that all $W'$-marginal cells of $Γ$ are
$W'$-outer-perimetric cells (corresponding to the last two cases of \autoref{label_rigoureusement}).

% \ig{in \autoref{label_simoneggiando}, the three leftmost green cells are somehow ambiguous: one may think that there are 1, 2, or 2 cells}

%{We finally define $C_{W'}(\Gamma)=C_{W'}^{\sf ip}(\Gamma)\cup C_{W'}^{\sf op}(\Gamma)\cup C_{W'}^{\sf in}(\Gamma)\cup C_{W'}^{\sf ex}(\Gamma).$}

%A cell of $\Gamma$ is \emph{$W'$-external} (resp. \emph{$W'$-internal}) if it is either $W'$-strictly or $W'$-outer-perimetric (resp. internal).
%We use the terms $C_{W'}^{\sf i}(\Gamma), C_{W'}^{\sf e}(\Gamma)$ for the $W'$-internal and the $W'$-external cells respectively.
\medskip

%Let $(\bar{W},\bar{\frak{R}})$ be a flatness pair of a graph $G,$ where $\frak{R}=(X,Y,P,C,\Gamma,\sigma,\pi).$
%An $W$-internal or $W$-inner-perimetric cell of ${\frR}$ is {\em untidy} if  $\pi(\tilde{c})$ contains a vertex
%$x$ of ${W}$ such that two of the edges of ${W}$ that are incident to $x$ are edges of $\sigma(c).$
%%
%% as the two edges incident to $x$ are edges of induced paths of $W$ inside $\sigma(c)$ leading to two, different than
%%$x,$ vertices of $\pi(\tilde{c}).$
%An $W$-internal or $W$-inner-perimetric  cell of  ${\frR}$ is {\em tidy} if it is not untidy.

\begin{figure}[ht]
	\begin{center}
		\includegraphics[width=15.6cm]{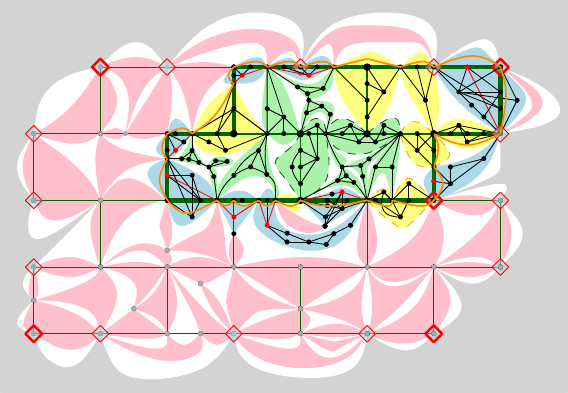}
	\end{center}		\vspace{-6mm}
	\caption{
		A flat wall $W$ in a graph $G,$ the painting of a rendition $\frak{R}$ certifying its flatness, a subwall $W'$ of $W,$ of height three, which is
		$\frak{R}$-normal, and the $\frak{R}$-flaps of $W,$ corresponding  to  the cells of $\frak{R}$ that are not $W'$-external.
		The edges and the non-boundary vertices of the flaps corresponding to the $W'$-external cells of ${\frak{R}}$ (depicted in pink) are not depicted (however their boundary vertices that are not in $D(W')$ are depicted in grey).
		There are nine $W'$-outer-perimetric cells  of ${\frak{R}}$ (in blue) and seven $W'$-inner-perimetric cells (in yellow). Also, there are thirteen $W'$-internal cells of ${\frak{R}}$ (in green). Among the $W'$-inner-perimetric and $W'$-internal cells of ${\frak{R}},$ those that are untidy are depicted with a dashed boundary. The orange cycle is the circle $K_{W'}.$
	}

	\label{label_simoneggiando}
\end{figure}
%
%\begin{lemma}
%\label{label_postmetaphysical}
%There is an algorithm that, given a graph $G$  and a flatness pair $(\bar{W},\bar{\frak{{R}}})$ of $G$ transforms it to
%a regular flatness pair  $({W},{\frak{{R}}})$ such that
%$W$ is a tilt of $\bar{W}$ and  $({W},{\frak{{R}}})$  is the same as  $(\bar{W},\bar{\frak{{R}}}),$ in case  $(\bar{W},\bar{\frak{{R}}})$ is regular.
%This algorithm runs in $O(n+m)$ time{, where $n$ (resp. $m$) is the number of vertices (edges) of \sed{Say this universally!}$G.$}\\
%Moreover, there is a $O(n+m)$ time computable function  $\theta,$\sed{Adapt the proof on this notation!}
%mapping each $W'\in {\cal S}_{\frak{R}}(W)$ to a flatness pair  of $G,$ such that, for every $W'\in {\cal S}(W),$ where $\theta(W')=(W'',\frak{R}''),$ the following hold:
%\begin{itemize}
%\item $(W'',\frak{R}'')$ is regular,\sed{indetity in the same?}
%\item $(W'',\frak{R}'')$ is a tilt pair of $W'$ in $(W,\frak{R}),$ and
%\item ${\sf compass}_{\frak{R}''}(W'')$ is a subgraph of ~${\sf influence}_{\frak{R}}(W').$
%\end{itemize}
%\end{lemma}

\subsection{The main lemma}
\label{label_desampararos}

\begin{lemma}
	\label{label_weltverbesserer}
	There is an algorithm that, given a graph $G,$ a  flatness pair $({W},{\frak{R}}),$ where  ${\frak{R}}=({X},{Y},{P},{C},{\Gamma},{\sigma},{\pi}),$ and a wall $W'\in{\cal S}_{{\frak{R}}}({W}),$ outputs,  in $\Ocal(n+m)$ time, a  flatness pair $(\tilde{W}',\tilde{\frak{R}}')$ where $\tilde{\frak{R}}'=(X',Y',P',C',\Gamma',\sigma',\pi')$ such that
	\begin{enumerate}
		%\item $(W'',\frak{R}'')$ is a tilt pair of $W'$ in $(W,\frak{R}),$
		\item all
		      cells of $\tilde{\frak{R}}'$ are $\tilde{W}'$-internal or $\tilde{W}'$-inner-perimetric,
		\item  $\tilde{W}'$ is a tilt of $W',$
		      %\item the set of internal cells of  $\frak{R}'$ is the same as the set of $W$-{internal}  cells of $\bar{\frak{R}},$
		\item $\sigma'|_{C_{\tilde{W}'}^{\sf in}(\Gamma')}={\sigma}|_{C_{W'}^{\sf in}({\Gamma})},$ i.e.,  the set of $\tilde{W}'$-internal
		      cells of  $\tilde{\frak{R}}'$ is the same as the set of $W'$-internal cells
		      of ${\frak{R}}$ and their images via ${\sigma'}$ and $\sigma$ are also the same, and
		\item ${\sf compass}_{\tilde{\frak{R}}'}(\tilde{W}')$ is a subgraph of $\cupall{\sf influence}_{{\frak{R}}}(W').$
	\end{enumerate}
	Moreover, if all  $W'$-internal or $W'$-inner-perimetric cells of ${\frak{R}}$ are tidy, then   the flatness pair $(\tilde{W}',\tilde{\frak{R}}')$ is regular.
\end{lemma}
%Finally notice that if all $D(\bar{W})$-internal  cells of $\Gamma $ are tidy

%PROOF REMOVED

\begin{proof}
	Since ${\frak{R}}=(X,Y,P,C,\Gamma,\sigma,\pi)$ is a  7-tuple certifying that ${W}$ is flat in $G$ , we have that the triple  $(\Gamma,\sigma,\pi)$ is an $\Omega$-rendition of $G[Y],$ where $\Gamma=(U,N)$ is a $\Delta$-painting.
	%Moreover, by~\autoref{label_levantadores}, we can also assume that $(\bar{W},\bar{\frak{R}})$ is tight.
	%We set $D:=D(W'),$ i.e., $D$ is the perimeter of $W'.$
	%

	We define a series of ingredients that will permit us to define an alternative  7-tuple $\tilde{\frak{R}}'.$
	As a first step, for every  $W'$-inner-perimetric cell $c\in C_{W'}^{\sf ip}(\Gamma)$ we define an arc $Y_{c}$ of $\Delta,$  as in  \autoref{label_affirmations} (where $Y_{c}$ is depicted in red), we set $F_{1}^{c}=\sigma(c),$ $r_{c}=1,$ and $V_{\rm mid}^{c}=\pi(\tilde{c})\cap V(D(W'))$ (the vertices in $V_{\rm mid}^{c}$ are depicted in orange in \autoref{label_affirmations}).
	%ΣΕΔ>\sed{The third mid shuld be excluded} but it does not mind as this will be included for some other cell

	\begin{figure}[h]
		\begin{center}
			\includegraphics[width=13cm]{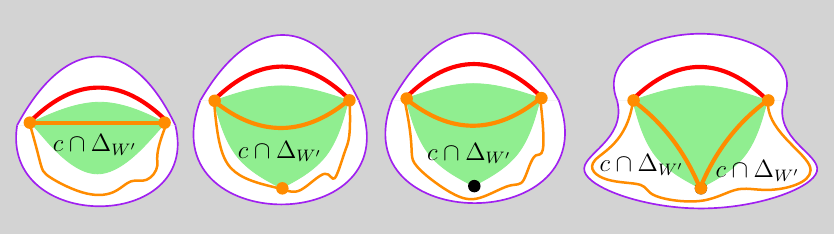}
		\end{center}		\vspace{-2mm}
		\caption{The four cases of the definition of the arc $Y_{c}$ (depicted in red), for  $W'$-inner-perimetric cells. The boundary of $\Delta_{W'}$ is depicted in orange and the boundary of  $\Delta$ is depicted in purple.}
		\label{label_affirmations}
	\end{figure}

	Next, we consider a  $W'$-outer-perimetric cell $c\in C_{W'}^{\sf op}(\Gamma).$ We assume that $\pi(\tilde{c})=\{x,y,z\}$ and that $x$ and $y$ are  the two endpoints of the non-trivial path of $D(W')\cap \sigma(c)$ (by {\em non-trivial} we refer  to the path that has distinct endpoints).
	We also define $V^{c}_{W'}$ as the set of all internal endpoints of this path that are different from $z.$
	Let  $\langle F_{1}^c,\ldots,F_{r_c}^c\rangle$ be the stretching
	%\sed{Shortest paths $O(n+m)$}
	of  $\sigma(c)$ along the pair $(x,y)$ and let $v_i,$ for $i\in[r_{c}-1],$  be the common endpoint of $F_{i}^{c}$ and $F_{i+1}^{c}.$
	Notice that by {tightness property $(i)$}, $r_c\geq 2.$ This permits us to set up a {\em special vertex} $v^{c}=v_1.$
	We also set
	\begin{align*}
		V_{\rm mid}^{c} & =  \{x,v_{1},\ldots,v_{r_c-1},y\},                                   &
		V_{\rm in}^{c}  & =  \cupall\{V(F_{i}^{c})\mid i\in[r_{c}]\}\setminus V_{\rm mid}^{c}.   %,\text{~and~}\\
		%V_{\rm out}^{c} & = & V(\sigma(e)) \setminus \cupall\{V(F_{i}^{c})\mid i\in[r_{c}]\}.
	\end{align*}

	Let $p_{0}=\pi^{-1}(x),$ $p_{r_c}=\pi^{-1}(y),$ and create a collection $c_{1},\ldots,c_{r_c}$ of open disks in ${c}$ and a set $p_{1},\ldots,p_{r_c-1}$ of points in $c$ such that
	\begin{itemize}
		\item
		      $p_0\in \bd(c_{1})$ and $p_{r_c}\in \bd(c_{r_c}),$  $p_{0}\neq p_{1},$ and  $p_{r_c}\neq p_{r_c-1},$
		\item for $i\in [r_c-1],$ $\bar{c}_{i}\cap \bar{c}_{i+1}=\{p_{i}\},$ and
		\item  for every $(i,j)\in{[r_c]\choose 2},$ $\bar{c}_{i}\cap \bar{c}_{j}\neq\emptyset$ if and only if  $|i-j|=1.$
	\end{itemize}

	We define {\em the cell replacement} of $c$ as the set   ${\sf c\mbox{-}repl}(c)=\{c_{1},\ldots,c_{r_c}\},$ {\em the point replacement} of $c$  as the set ${\sf p\mbox{-}repl}(c)=\{p_{0},\ldots,p_{r_{c}}\},$ and we set $C_{\rm new}^{c}  =  \cupall  {\sf c\mbox{-}repl}(c)$ and $N_{\rm new}^{c} =  \cupall  {\sf p\mbox{-}repl}(c).$

	We also define the arc $Y_{c}$ as an arc of $c$ where $p_{i}\in Y_{c}, i\in[0,r_c],$ such that $p_{0},$ $p_{r_c}$ are the {extreme points} of $Y_{c},$ and  $Y_{c}$ is traversing $\tilde{c}$ as depicted by the red line in \autoref{label_indistinguishable}.
	Observe  that  $\cupall\{Y_{c}\mid c\in  C_{W'}^{\sf ip}(\Gamma)\cup C_{W'}^{\sf op}(\Gamma)\}$
	is a ``red'' cycle of $\Delta.$ Let $\Delta'$ be the disk bounded by this cycle for which $\Delta'\subseteq \Delta.$
	\begin{figure}[h]
		\begin{center}
			\includegraphics[width=13cm]{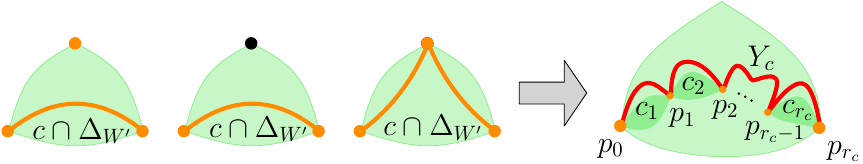}
		\end{center}\vspace{-2mm}
		\caption{The  definition of the replacement sequence $c_{1},\ldots,c_{r_{c}}$ and the arc $Y_{c}$ for  the {three}  cases of $W'$-external cells of $C_{W'}^{\sf op}(\Gamma).$}
		\label{label_indistinguishable}
	\end{figure}

	We  set
	\begin{align*}
		H           & =   \cupall\{F_{1}^c\cup\cdots\cup F^c_{r_{c}}\mid  c\in   C_{W'}^{\sf op}(\Gamma)\},           &
		V_{W'}      & =  \cupall\{V^{c}_{W'}\mid c\in C_{W'}^{\sf op}(\Gamma)\},                                        \\
		V_{\rm mid} & =     \cupall\{V_{\rm mid}^{c} \mid c\in C_{W'}^{\sf ip}(\Gamma)\cup C_{W'}^{\sf op}(\Gamma)\}, &
		V_{\rm in}  & =     \cupall\{V_{\rm in}^{c} \mid c\in C_{W'}^{\sf op}(\Gamma)\},                                \\
		N_{\rm new} & =   \cupall\{N_{\rm new}^c\mid c\in C_{W'}^{\sf op}(\Gamma)\},                                  &
		U_{\rm new} & =   \cupall\{C_{\rm new}^c\cup N_{\rm new}^c\mid c\in C_{W'}^{\sf op}(\Gamma)\}.
	\end{align*}
	We now define the wall $\tilde{W}'=(W'\setminus V_{W'})\cup H,$ i.e., we extract from $W'$ the internal vertices of the subpaths
	of $W'$ that are intersected by images, via $\sigma,$ of $W'$-outer-perimetric cells and we substitute them by the paths of their stretchings. Clearly this does not affect the interior of $W',$ and therefore $\tilde{W}'$ is a tilt of $W',$ yielding Property {\em 2} of the statement of the lemma.
	Next we define a separation $(X',Y')$ of $G$ so that
	\begin{align*}
		Y' & =  \cupall\{V(\sigma(c))\mid c\in C_{W'}^{\sf ip}(\Gamma)\cup C_{W'}^{\sf in}(\Gamma)\}\cup V_{\rm in} \cup V_{\rm mid}, &
		X' & =  (V(G)\setminus Y')\cup V_{\rm mid}.
	\end{align*}
	In other words, $Y'$ consists of the images of the internal cells and the vertices of every path $F_{i}^{c},$
	while $X'$ consists of everything else, except from $V_{\rm mid}$ (that is, the set $X'\cap Y'$). Notice that
	\begin{eqnarray}
		G[Y']\mbox{~ is a subgraph of~~}  \cupall\{\sigma(c)\mid c\in  C_{W'}^{\sf in}(\Gamma)\cup C_{W'}^{\sf ip}(\Gamma)\cup  C_{W'}^{\sf op}(\Gamma)\}={\sf influence}_{{\frak{R}}}(W'). \label{label_circonstance}
	\end{eqnarray}

	We define the pair $(P',C')$ as follows.
	Let $c$ be a $W'$-outer-perimetric cell and  $\sigma(c)\cap V(D(W'))$ contain  a vertex $w$
	such that either $w$ is a 3-branch vertex of  $W'$
	or  $w\in P$ (resp. $w\in C$). We distinguish
	two cases. If $w\in Y',$ then we include $w$ in $P'$ (resp. $C'$). If  $w\not\in Y',$ then we include the special vertex $v^{c}$  in $P'$ (resp. $C'$).

	We next define an  $\Omega'$-rendition $(\Gamma',\sigma',\pi')$ of $G[Y']$ where $\Gamma'=(U',N')$ is a $\Delta'$-painting.
	For this we set  $\Gamma'=(U',N'),$ where
	\begin{eqnarray*}
		U'  =  \big(\big(U\setminus\cupall C_{W'}^{\sf op}(\Gamma) \big)\cap \Delta'\big) \cup U_{\rm new} &\ \text{ and }\ &
		N'  =   (N\cap  \Delta')\cup N_{\rm new}.
	\end{eqnarray*}
	Let now $K'$ be the set of the connected components of $U'\setminus N',$ which will form the cells of the new  $\Omega'$-rendition $(\Gamma',\sigma',\pi').$
	We define the function $\sigma'$ mapping the cells in $C'$ to subgraphs of $G[Y']$
	as follows. Notice that $c\in K'\cap  C(\Gamma)$ if and only if
	$c\in C_{W'}^{\sf in}(\Gamma)\cap C_{W'}^{\sf ip}(\Gamma),$ and in this case we set $\sigma'(c)=\sigma(c).$
	Suppose now that $c\in K'\setminus   C(\Gamma).$ Then
	$c$ should be one of the cells, say $c_{i},$ of ${\sf c\mbox{-}repl}(c^*)=\{c_{1},\ldots,c_{r_c}\}$ for some $c^*\in C_{W'}^{\sf op}(\Gamma),$ and in this case we set $\sigma(c)=F_{i}^{c^*}.$ It now remains to define $\pi':N'\to Y'.$ Similarly to the definition of $\sigma',$ we consider a $p'\in N'$ and if $p\in N\cap N'$  we set $\pi'(p)=\pi(p).$ Suppose now that $p\in N'\setminus  N.$ Then
	$p$ should be one of the points, say $p_{i},$ of ${\sf p\mbox{-}repl}(c^*)=\{p_{0},\ldots,p_{r_c}\}$ for some $c^*\in C_{W'}^{\sf op}(\Gamma)$ and such that $i\in[r_{c^*}-1].$ In this case we define  $\pi'(p)$ to be the unique common vertex of $F_{i}^{c^*}$ and $F_{i+1}^{c^*}.$ It is now easy to verify that  $(\Gamma',\sigma',\pi')$ is a tight $\Omega'$-rendition of $G[Y']$ and  that the 7-tuple $\tilde{\frak{R}}':=(X',Y',P',C',\Gamma',\sigma',\pi')$ certifies that $\tilde{W}'$ is flat in $G$ (see \autoref{label_hohenzollern}). Moreover $K'=C(\Gamma').$

	\begin{figure}[h]
		\begin{center}
			\includegraphics[width=14cm]{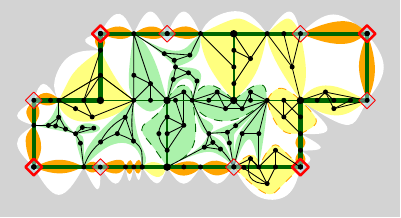}
		\end{center}		\vspace{-2mm}
		\caption{The flatness pair  $(\tilde{W}',\tilde{\frak{R}}')$ created in the proof of \autoref{label_weltverbesserer}. The wall $\tilde{W}'$ is the tilt of $W'$  where the updated part of $\tilde{W}'$
			correspond to the red paths in  \autoref{label_simoneggiando} whose edges are drawn in the orange cells.}
		\label{label_hohenzollern}
	\end{figure}

	Recall now that all the cells in   $C(\Gamma')\cap  C(\Gamma)$ are either $\tilde{W}'$-inner-perimetric or $\tilde{W}'$-internal.
	Moreover, all  the cells in $C(\Gamma')\setminus   C(\Gamma)$ are cells as in the left part of \autoref{label_rigoureusement}, therefore they are $\tilde{W}'$-inner-perimetric. This yields Property {\em 1} in the statement of the lemma. Notice also that Property {\em 3} follows directly from the definition of $\sigma',$ as
	it concerns the $W'$-internal cells of $\frak{R},$ and these cells are the same as the
	$\tilde{W}'$-internal cells of $\tilde{\frak{R}}'.$
	Finally, recall  that ${\sf compass}_{\tilde{\frak{R}}'}(\tilde{W}')=G[Y']$ and
	Property {\em 4} follows because of~\eqref{label_circonstance}.

	On the other hand, notice that all $\tilde{W}'$-internal cells of $\tilde{\frak{R}}'$ are also
	$W$-internal cells of ${\frak{R}}.$ Moreover, if a  $\tilde{W}'$-inner-perimetric cell $c$ of $\tilde{\frak{R}}'$ is a cell of ${\frak{R}},$ then $c$ is either
	an ${W}$-inner-perimetric  or an  ${W}$-internal cell of ${\frak{R}}.$ On the other hand, all  $\tilde{W}'$-inner perimetric cells of $\tilde{\frak{R}}'$ that are
	not cells of ${\frak{R}}$ are cells as in the left part of \autoref{label_rigoureusement}, therefore they are $\tilde{W}'$-inner-perimetric and tidy.
	We conclude that if all  $W'$-internal or $W'$-inner-perimetric  cells of ${\frak{R}}$ are tidy, then  all cells of $\tilde{\frak{R}}'$ are tidy as well. As $\tilde{\frak{R}}'$ does not have any  $\tilde{W}'$-outer-perimetric cells it also does not have $\tilde{W}'$-marginal cells. These two facts along with the fact that
	$\tilde{\frak{R}}'$ does not have any  $\tilde{W}'$-external cells imply that  the flatness pair $(\tilde{W}',\tilde{\frak{R}}')$ is regular.

	The running time follows from the fact that the substitution of $W'$-outer-perimetric cells is based on the stretching operation on the corresponding flaps, and this requires the computation of shortest paths that, in total, takes $\Ocal(n+m)$ time.
\end{proof}

\begin{lemma}
	\label{label_meretricious}
	There is an algorithm that, given a graph $G$ and a  flatness pair $({W},\frak{{R}}),$ outputs, in $\Ocal(n+m)$ time, a  flatness pair $(W^{\star},\frak{R}^{\star})$ of $G$ {with the same height as $({W},\frak{{R}})$}, with $\frak{R}^{\star}=\frak{R},$ and such that all  the $W^{\star}$-internal or $W^{\star}$-inner-perimetric cells  of $\frak{R}^{\star}$ are tidy.
\end{lemma}
\begin{figure}[ht]
	\begin{center}
		\includegraphics[width=14cm]{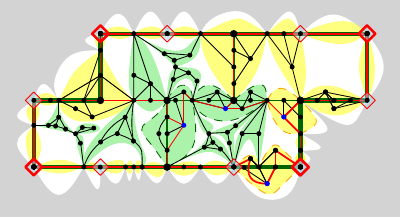}
	\end{center}		\vspace{-2mm}
	\caption{An illustration of the proof of \autoref{label_meretricious}, based on the flatness pair of \autoref{label_hohenzollern}.} The new flatness pair
	is $({W}^{\star},\frak{{R}}^{\star})$ where $W^{\star}$ is  depicted in red and $\frak{{R}}^{\star}=\frak{{R}}.$
	\label{label_econonomically}
\end{figure}

\begin{proof}
	Given a wall $W$ and an  $\frR=(X,Y,P,C,\Gamma,\sigma,\pi)$ as above, we
	denote by $C_{W}^{\sf utd}(\Gamma)$  the set of all the $W$-internal or $W$-inner-perimetric cells of $\Gamma$ that are untidy.
	Notice that for every $c\in C_{W}^{\sf utd}(\Gamma),$ $|\pi(\tilde{c})|=3.$
	In what follows, we explain how to update $W,$  while leaving $(X,Y,P,C,\Gamma,\sigma,\pi)$ intact,
	in order to reduce $|C_{W}^{\sf utd}(\Gamma)|$ by one.  Repeating this procedure clearly yields the statement claimed in the lemma.

	Let $c\in C_{W}^{\sf utd}(\Gamma).$ We assume that $\pi(\tilde{c})=\{x,y,z\}$ and that $z\in \pi(\tilde{c})\cap V(W)$ is a  vertex of $W$ such that two of the edges of $W$ incident to $z$ are edges of $\sigma(c).$ This implies that $\bar{P}=W\cap \sigma(c)$ is an $(x,y)$-path
	containing $z$ as an internal vertex. Moreover, none of the internal vertices of $\bar{P},$ except from $z,$ is a 3-branch vertex of $W.$  By {tightness properties (i), (ii), and (iii)}, there is a vertex $w\in \sigma(c)\setminus \pi(\tilde{c})$ and three internally vertex-disjoint paths $P_{x}',$ $P_{y}',$ and $P_{z}'$ in $\sigma(c)$ such that $P_{x}'$ is a $(w,x)$-path,
	$P_{y}'$ is a $(w,y)$-path, and  $P_{z}'$ is a $(w,z)$-path. If $z$ is a 3-branch vertex of $W$ we update $W:=(W\setminus V(\bar{P}\setminus \{x,y,z\}))\cup P_{x}'\cup P_{y}'\cup P_{z}'$ {(see bottom yellow cell with dashed boundary in \autoref{label_econonomically} for an example)}, while, if not, we
	update $W:=(W\setminus V(\bar{P}\setminus \{x,y\}))\cup P_{x}'\cup P_{y}'$ {(see the leftmost green cell with dashed boundary in \autoref{label_econonomically} for an example)}
	and observe that $W$ is again a flat wall of $G,$ certified by $(X,Y,P,C,\Gamma,\sigma,\pi).$ Moreover, in the first case, $z$ is not anymore a 3-branch vertex of $W$ and is incident to only one edge of $\sigma(c)\cap W,$ while, in the second case, $z$ is not anymore a  vertex of $W.$ This implies that $c$ is not anymore untidy and $|C_{W}^{\sf utd}(\Gamma)|$ is indeed reduced by one (see \autoref{label_econonomically} for an example).
	As for each cell $c$ that we modify we need to identify the paths $P_{x}',$ $P_{y}',$ and $P_{z}'$ in $\sigma(c),$ the construction of $W'$ takes, in total, $\Ocal(n+m)$ time.
\end{proof}

%\remove{
%\blue{
%
%}
%}
%\hrule

\subsection{Proofs of  \autoref{label_proporcionada} and  \autoref{label_considerabil}}
\label{label_darstellungsmittel}

We finally have all the ingredients to prove our two main results.

\begin{proof}[Proof of \autoref{label_proporcionada}]
	Let $({W},{\frak{R}})$ be a flatness pair  of a graph $G,$ where  ${\frak{R}}=({X},{Y},{P},{C},{\Gamma},{\sigma},{\pi})$ and $W'\in{\cal S}_{{\frak{R}}}({W}).$
	We call the algorithm of \autoref{label_weltverbesserer} on  $G,$ $({W},{\frak{R}}),$ and $W',$ which outputs,  in $\Ocal(n+m)$ time, a  flatness pair $(\tilde{W}',\tilde{\frak{R}}')$ where $\tilde{\frak{R}}'=(X',Y',P',C',\Gamma',\sigma',\pi')$ such that all
	cells of $\tilde{\frak{R}}'$ are $\tilde{W}'$-internal or $\tilde{W}'$-inner-perimetric (hence $\tilde{\frak{R}}'$ does not have $\tilde{W}'$-external cells),  $\tilde{W}'$ is a tilt of $W',$ the set of $\tilde{W}'$-internal
	cells of  $\tilde{\frak{R}}'$ is the same as the set of $W'$-internal cells
	of ${\frak{R}}$ and their images via ${\sigma'}$ and $\sigma$ are also the same, and ${\sf compass}_{\tilde{\frak{R}}'}(\tilde{W}')$ is a subgraph of $\cupall{\sf influence}_{{\frak{R}}}(W').$
	We observe that $(\tilde{W}',\tilde{\frak{R}}')$ is a $W'$-tilt of $(W,\frak{R})$ and thus we return $(\tilde{W}',\tilde{\frak{R}}').$
	Notice that in the case where $({W},{\frak{R}})$ is regular, all cells of $\frak{R}$ are tidy.
	Thus, by \autoref{label_weltverbesserer}, $(\tilde{W}',\tilde{\frak{R}}')$ is also regular.%Finally,  we have to prove that, if $({W},\frak{R})$ is regular, then the outputs of $\theta$ are regular as well. \ig{in order to apply the last sentence of \autoref{label_weltverbesserer}, we need that all $W$-internal or $W$-inner-perimetric cells of $\frak{R}$ are tidy, which a priori is \textbf{not} implied by regularity. Why is it the case?  I would like to avoid applying \autoref{label_meretricious}, since it changes the wall, so we lose the fact that the wall generated by $\theta$ is a tilt of the original wall, which is required in a tilt-assignment function. I assume that this property holds in the remainder of this proof} Given $W'\in{\cal S}_{{\frak{R}}}({W}),$ since all $W$-internal or $W$-inner-perimetric cells of $\frak{R}$ are tidy, by \autoref{label_weltverbesserer} all ($\tilde{W}'$-internal or ${W}'$-inner-perimetric) cells of $\tilde{\frak{R}}'$ are tidy, and \autoref{label_indisdinctively} implies that $\theta(W')= (\tilde{W}',\tilde{\frak{R}}')$ is regular.
\end{proof}

%\autoref{label_proporcionada} follows from \autoref{label_weltverbesserer} and \autoref{label_considerabil} follows from
%\autoref{label_weltverbesserer}, \autoref{label_meretricious}, and \autoref{label_indisdinctively}.

\begin{proof}[Proof of \autoref{label_considerabil}]
	Given a flatness pair  $({W},{\frak{R}})$ of a graph $G,$ we first apply
	\autoref{label_meretricious} to  $(W,\frak{R})$ and obtain in time $\Ocal(n+m)$ a  flatness pair $(\hat{W}^{\star},\hat{\frak{R}}^{\star})$ of $G$ with the same height as $({W},\frak{{R}}),$ with $\hat{\frak{R}}^{\star}=\frak{R},$ and such that all  $\hat{W}^{\star}$-internal or $\hat{W}^{\star}$-inner-perimetric cells  of $\hat{\frak{R}}^{\star}$ are tidy.

	We now apply \autoref{label_weltverbesserer} with input $G,$ $(\hat{W}^{\star},\hat{\frak{R}}^{\star}),$ and $\hat{W}^{\star}$ and obtain, in ${\cal O}(n+m)$ time, a flatness pair $({W}^{\star},{\frak{R}}^{\star})$ of $G$ such that, if $\hat{\frak{R}}^{\star}=(\hat{X},\hat{Y},\hat{P},\hat{C},\hat{\Gamma},\hat{\sigma},\hat{\pi})$ and $\frak{R}^{\star}=(X,Y,P,C,\Gamma,\sigma,\pi),$ we have that all
	cells of $\frak{R}^{\star}$ are $W^{\star}$-internal or $W^{\star}$-inner-perimetric (hence $\frak{R}^{\star}$ does not have ${W}^{\star}$-external cells),  $W^{\star}$ is a tilt of $\hat{W}^{\star},$ the set of $W^{\star}$-internal
	cells of $\hat{W}^{\star}$ is the same as the set of $\hat{W}^{\star}$-internal cells
	of $\hat{\frak{R}}^{\star}$ and their images via ${\sigma}$ and $\hat{\sigma}$ are also the same, and ${\sf compass}_{\frak{R}^{\star}}({W}^{\star})$ is a subgraph of $\cupall{\sf influence}_{\hat{\frak{R}}^{\star}}(\hat{W}^{\star}).$
	Moreover, since  all  the $\hat{W}^{\star}$-internal or $\hat{W}^{\star}$-inner-perimetric cells  of $\hat{\frak{R}}^{\star}$ are tidy, \autoref{label_weltverbesserer} implies that  all ($W^{\star}$-internal or $W^{\star}$-inner-perimetric)  cells of $\frak{R}^{\star}$ are tidy.
	Also, since none of the cells of $\frak{R}^{\star}$ is $W^{\star}$-outer-perimetric, none of the cells of $\frak{R}^{\star}$ is $W^{\star}$-marginal. These two facts together with the fact that none of the cells of $\frak{R}^{\star}$ is $W^{\star}$-external imply that $({W}^{\star},{\frak{R}}^{\star})$ is a regular flatness pair of $G$ with the same height as $({W},{\frak{R}}),$ as required.

	We now prove that ${\sf compass}_{{\frak{R}}^{\star}}({W}^{\star})\subseteq {\sf compass}_{\frak{R}}(W).$ First, keep in mind that ${\sf compass}_{\frak{R}^{\star}}({W}^{\star})\subseteq \cupall{\sf influence}_{\hat{\frak{R}}^{\star}}(\hat{W}^{\star}).$
	We observe that $\cupall{\sf influence}_{\hat{\frak{R}}^{\star}}(\hat{W}^{\star})\subseteq {\sf compass}_{\hat{\frak{R}}^{\star}}(\hat{W}^{\star})$ and, since $\hat{\frak{R}}^{\star}=\frak{R},$ ${\sf compass}_{\hat{\frak{R}}^{\star}}(\hat{W}^{\star})={\sf compass}_{\frak{R}}(W).$ Therefore, ${\sf compass}_{\frak{R}^{\star}}({W}^{\star})\subseteq {\sf compass}_{\frak{R}}(W).$

	Finally, the claimed running time follows from \autoref{label_weltverbesserer} and \autoref{label_meretricious}.
	%We now apply
	%\autoref{label_weltverbesserer} with input $G,$ $(W^{\star},\frak{R}^{\star}),$ and $W^{\star},$
	%and obtain in time $\Ocal(n+m)$ a flatness pair $(\tilde{W^{\star}},\tilde{\frak{R}^{\star}})$ of $G$ such that, in particular, all
	%cells of $\tilde{\frak{R}^{\star}}$ are tidy and either $\tilde{W}^{\star}$-internal or $\tilde{W}^{\star}$-inner-perimetric, and  $\tilde{W}^{\star}$ is a tilt of $W^{\star},$ hence $(\tilde{W}^{\star},\tilde{\frak{R}}^{\star})$ and $({W^{\star}},{\frak{R}^{\star}})$ have the same height, and thus $(\tilde{W}^{\star},\tilde{\frak{R}}^{\star})$ and $({W},{\frak{R}})$ have the same height as well. Finally, \autoref{label_indisdinctively} implies that $(\tilde{W}^{\star},\tilde{\frak{R}}^{\star})$ is a regular flatness pair of $G,$ which  admits an $\Ocal(n+m)$-computable tilt-assignment function $\theta$ by \autoref{label_proporcionada}.
\end{proof}

%\ig{in the references, better to put DOI than URL}

\bibliographystyle{plainurl}

\end{document}